\documentclass[11pt]{article}
\usepackage{enumitem}

\usepackage{authblk}
\usepackage[dvipsnames]{xcolor}
\usepackage{wrapfig}

\usepackage[mathscr]{eucal}
\usepackage{hyperref}
\usepackage[utf8]{inputenc}
\usepackage{graphicx}
\usepackage{mathtools}
\usepackage{amssymb}
\usepackage{amsmath}
\usepackage{amsthm}
\usepackage{bm}
\usepackage[normalem]{ulem}

\newtheorem{theorem}{Theorem}
\newtheorem{lemma}[theorem]{Lemma}
\newtheorem{proposition}[theorem]{Proposition}
\newtheorem{corollary}[theorem]{Corollary}

\theoremstyle{definition}
\newtheorem{definition}[theorem]{Definition}
\theoremstyle{remark}
\newtheorem{remark}[theorem]{Remark}

\newcommand{\Na}{N_{\alpha}}

\newcommand{\Nb}{N_{\beta}}
\newcommand{\pa}{p_{\alpha}}
\newcommand{\pb}{p_{\beta}}

\newcommand{\Nmax}{N_{\textrm{max}}}
\newcommand{\Iset}{\mathcal{I}}

\newcommand{\ia}{\iota_{\alpha}}
\newcommand{\ib}{\iota_{\beta}}
\newcommand{\psiab}{\psi_{\alpha\beta}}
\newcommand{\psiba}{\psi_{\beta\alpha}}

\newcommand{\ima}{\iota_{\textrm{max}}}
\newcommand{\rec}{R_{\varepsilon_{0},\varepsilon_{1}}}

\newcommand{\ga}{\gamma_{1}}
\newcommand{\gb}{\gamma_{2}}

\newcommand{\qb}{q_{\beta}}

\renewcommand{\r}{\rho}
\newcommand{\Jset}{\mathcal{J}}
\newcommand{\R}{\mathbb{R}}
\newcommand{\Nupp}{N_{\textrm{upp}}}
\newcommand{\Mextone}{M_{\textrm{ext},1}}
\newcommand{\Mexttwo}{M_{\textrm{ext},2}}
\newcommand{\iupp}{\iota_{\textrm{upp}}}

\newcommand{\pr}{\partial}
\newcommand{\mext}{M_{\textrm{ext}}}
\newcommand{\gext}{g_{\textrm{ext}}}

\usepackage{vmargin} 
\setmargins{2.5cm}{1.5cm}{16.5cm}{23.42cm}{10pt}{1cm}{0pt}{2cm}

\date{\today}

\title{Uniqueness of maximal spacetime boundaries}

\author{Melanie Graf\thanks{Department of Mathematics, Universit\"at Hamburg,  20146 Hamburg, Germany; \url{melanie.graf@uni-hamburg.de} \\
Significant parts of this work was completed while at the Department of Mathematics, Universit\"at T\"ubingen, 72076 T\"{u}bingen, Germany and at the Institute for Mathematics, Universit\"at Potsdam, 14476 Potsdam, Germany.} \;
and Marco van den Beld-Serrano\thanks{Department of Mathematics, Universit\"at Regensburg, 93053 Regensburg, Gemany; \url{marco.van-den-beld-serrano@ur.de}. Significant parts of this work were completed while at the Department of Mathematics, Universit\"at T\"ubingen, 72076 T\"{u}bingen, Germany.}}

\begin{document}

\date{}
\maketitle

\vspace{-.3in}

\begin{abstract}

Given an extendible spacetime one may ask how much, if any, uniqueness can in general be expected of the extension. 
Locally, this question was considered and comprehensively answered in a recent paper of Sbierski \cite{Sbierski2022}, where he obtains local uniqueness results for anchored spacetime extensions of similar character to earlier work for conformal boundaries by Chru\'sciel \cite{Chrusciel}.
Globally, it is known that non-uniqueness can arise from timelike geodesics behaving pathologically in the sense that there exist points along two distinct timelike geodesics which become arbitrarily close to each other interspersed with points which do not approach each other. 
We show that this is in some sense the only obstruction to uniqueness of maximal future boundaries: Working with extensions that are manifolds with boundary we prove that, under suitable assumptions on the 
regularity of the considered extensions and excluding the existence of such ''intertwined timelike geodesics'', 
extendible spacetimes admit a unique maximal future boundary extension. This is analogous to results of Chru\'sciel for the conformal boundary. 

\medskip

\end{abstract}

\setcounter{tocdepth}{1}
\tableofcontents

\section{Introduction}

Questions of (low-regularity) spacetime (in-)extendibility have a long history within mathematical general relativity and are closely related to several important physical problems such as the nature of the incompleteness predicted from the singularity theorems and strong cosmic censorship. The former has lead people to consider various ways of defining a boundary of spacetime (and attaching such boundaries to spacetime). As we will see, some of these old constructions are now providing useful inspirations, tools and reality checks in investigating uniqueness questions. The latter has of course been crucial motivation in studying low-regularity (in-)extendibility theory from the beginning in the hopes that the usually very general results developed in this field might provide useful additions to more PDE based approaches. 

In this general framework the usual procedure for determining whether a concrete spacetime or concrete class of spacetimes is extendible admits an extension $\iota:(M,g)\to (\mext,\gext)$ (with $(\mext,\gext)$ being a spacetime and $\iota$ an isometric embedding) or not is to follow one of two paths: either an explicit extension of the spacetime is found/constructed or it is shown that the spacetime satisfies some criteria that are known to be general obstructions to extendibility within a certain class of extensions.  For instance, blow up of any curvature scalar (e.g., the scalar curvature or the Kretschmann scalar) is an immediate obstruction to $C^{1,1}$-extendibility, that is there cannot exist a proper extension with $\gext\in C^{1,1}$.  
However, different strategies are required in order to explore the inextendibility of a spacetime in a lower regularity class (e.g $C^{0}$- or $C^{0,1}$-regularity). Here a lot of new tools and techniques have been developed in the last six years, leading to several nice results. For example, the question of $C^{0}$-inextendibility was first
tackled by Sbierski \cite{Sbierski2016}, who proved that the Minkowski and the maximally extended Schwarzschild spacetime are $C^{0}$-inextendible. We now have a collection of low regularity inextendibility criteria foremost amongst them timelike geodesic completeness: In the first place, in \cite{Sbierski2016} it was proven that if no timelike curve intersects the boundary of $M$ in the extension, $\pr \iota(M)$, then the spacetime is inextendible. This result already pointed to the idea that, under certain additional assumptions, timelike (geodesic) completeness would yield the inextendibility of a spacetime (in a low regularity class). Indeed, in \cite{LingGalSbierski} it was proven that a smooth globally hyperbolic and timelike geodesically complete spacetime is $C^0$-inextendible.\footnote{This result was later refined in several works: in \cite{Graf}, it was shown that if the global hyperbolicity condition is dropped the spacetime is at least $C^{0,1}$-inextendible. In a follow-up by Minguzzi and Suhr \cite{Minguzzi2019} it was shown that the global hyperbolicity condition can be dropped entirely and any smooth timelike geodesically complete spacetime must be $C^{0}$-inextendible and that a similar result holds in the Lorentz-Finsler setting. Finally in \cite{GrantKunzingerSaemann} an inextendibility result for timelike complete Lorentzian length spaces is established.} More importantly for us, \cite{GallowayLing} also showed that if the past boundary, $\pr^-\iota(M)$, is empty, then the future boundary, $\pr^+\iota(M)$, has to be an achronal topological hypersurface. This is a bit more generally applicable as often the behaviour to the past (or future) is better understood and there are several spacetimes, especially when looking towards cosmological models, that are future or past timelike geodesically complete but not both. Together with a structure result on the existence of certain nice coordinates around any boundary point by Sbierski (cf. Proposition \ref{Sbierskifuturechart}, this leads one to suspect that if $M$ is extendible but the past boundary is empty, $\iota(M)\cup \pr^+\iota(M)$ should be a topological manifold with boundary and, as we will discuss in Section \ref{sec:compatibility} indeed this is the case).

Surprisingly, in case $(M,g)$ is an arbitrary extendible spacetime, the general (i.e., without imposing additional symmetry, field equations or any strong regularity) question of uniqueness of extensions appears to have only recently come up, despite it being a very natural one.

Sbierski \cite{Sbierski2022} proved the local uniqueness of  $C^{0,1}_{\textrm{loc}}$-extensions up to (and including) the boundary in the following sense: Let $(M,g)$ be a globally hyperbolic spacetime and consider two $C^{0,1}_{\textrm{loc}}$-extensions $\iota_{1}$ and $\iota_{2}$ satisfying that there exists a future directed timelike curve $\gamma : [0,1) \rightarrow M$ (also called the \emph{anchoring curve}) such that $\iota_{1}\circ \gamma$ has a limit point $p_{1}\in\pr\iota_{1}(M)$ and $\iota_{2}\circ \gamma$ a limit point $p_{2}\in\pr\iota_{2}(M)$ as $t\to1$. Then, there exist suitable open subsets $U_1$ of $\iota_1(M)$ and $U_2$ of $\iota_2(M)$ containing $\iota_1\circ\gamma$ and $\iota_2 \circ \gamma $ such that the restriction of the identification map $\textrm{id}\coloneqq \iota_{1}\circ\iota_{2}^{-1}$  to these subsets extends to a $C^{1,1}_{\textrm{loc}}$-isometric diffeomorphism $\textrm{id} : U_{1}\cup (\pr\iota_{1}(M) \cap \pr U_1 ) \rightarrow U_{2}\cup (\pr\iota_{2}(M) \cap \pr U_2 ) $.
Hence, this implies local uniqueness of $C^{0,1}_{\textrm{loc}}$ extensions that 'extend through the same region'. These statements are nicely analogous to earlier local uniqueness results for conformal boundaries by Chru\'sciel \cite{Chrusciel}, albeit the details of the proofs clearly differ due to the different setting and the lower regularities Sbierski considers. 
 Sbierski also provides explicit examples that this local uniqueness fails if one allows extensions which are no longer $C^{0,1}_{\textrm{loc}}$. Once one has local uniqueness, the next natural question is if there is a sensible notion of 'maximal extension' and whether such maximal extensions may be globally unique in some sense.

In this paper we aim to answer these questions. However our setup is (out of necessity for our methods but also because of general considerations, cf.\ the discussion in Remark \ref{rem:whyboundary}) 
 a bit different from the classical spacetime extensions as we really focus on the boundary and on future directed timelike geodesics.
 This leads us to consider a different type of extensions of $M$ having the following properties:
\begin{enumerate}[label=(\roman*)]
    \item First, we consider a class of extensions in which the 'extended' manifold is a topological manifold with boundary. 
    \item Secondly, the 'extended' manifolds we work with can be seen as the result of 'attaching' to the original spacetime $M$ the limit points of inextendible incomplete (in $M$) timelike geodesics. That is, every point in the boundary should be the endpoint of a future directed timelike geodesic. 
 Further, we need to keep tight control on the topology of the extension at the boundary points. This is achieved by demanding that the manifold topology  of the extension can be reconstructed in a very precise way from the timelike geodesics of the original spacetime. 
 This description of a topology via so-called 'timelike thickenings' (see Definition \ref{def:thickenings})  is reminiscent of the old g-boundary construction by Geroch (see \cite{Gerochgboundary}) and further motivated by an analogous use of 'null thickenings' in Chru\'sciel's \cite{Chrusciel} work on maximal conformal boundaries.

    \item Third, sets of the form $\iota(M)\cup \pr^+\iota(M)$ should furnish examples of these new ''future boundary extensions'' -- at least for well behaved spacetime extensions $(\mext,\gext)$. We show that this is indeed the case if $(M,g)$ is globally hyperbolic, the past boundary of $(\mext,\gext)$ is empty and $\gext$ is $C^2$ in Section \ref{sec:compatibility}. In particular, whether $\iota(M)\cup \pr^+\iota(M)$ satisfies point two appears to be closely tied to the regularity of $\gext$: It should still work for $g\in C^{1,1}$, but becomes quite doubtful below that threshold. One may thus interpret (ii) as a regularity condition.
\end{enumerate}
We call these types of extensions \emph{regular future g-boundary extensions} and refer to Definition \ref{def:regular-g-boundary-ext} for the exact definitions. We will further motivate this definition in Section \ref{sec:maximalboundary}.  Our main goal will be to construct a \emph{unique maximal} regular future g-boundary extension (provided any such extension exists in the first place), where uniqueness is in the sense of the equivalence in Definition \ref{def:equiv of Ns}, i.e., the composition of the associated embeddings extends to a homeomorphism of topological manifolds with boundary. Note that our regular future g-boundary extensions do not come with a concept of extension of the metric to the boundary, so at this point our uniqueness really is topological in nature and we in particular don't claim anything about uniqueness of the metric on the boundary. This also means that we cannot use the metric at the boundary for our proofs, contrary to our main inspirations of \cite{Chrusciel, Sbierski2022}. However, in case there were a way of extending the metric to the boundary one might be able to combine our result with techniques from Sbierski's local results to obtain uniqueness of the metric on the boundary as well, but this would have to be explored in some future work. 
Another avenue for further exploration is that, except for the compatibility results in Section \ref{sec:compatibility}, we at this point do not investigate under which criteria given spacetimes possess a regular future g-boundary extension. This question would lead back to the general question of spacetime boundary constructions based on attaching endpoints to incomplete geodesics which generally are rather ill behaved topologically even when excluding the obvious potential offender of 'intertwined' timelike geodesics, that is roughly geodesics which never separate nor remain arbitrarily close as their affine parameter approaches the limit of their interval of existence (Definition \ref{defintertwined}), 
as an old example in \cite{Geroch2} shows.

\paragraph{Outline of the paper} We start by motivating and giving the definition for a regular future g-boundary extension in Section \ref{sec:maximalboundary} and discussing its relation with the usual concept of spacetime extensions in Section \ref{sec:compatibility}. Our procedure to construct a unique maximal regular future g-boundary extension, assuming that at least one regular future g-boundary extension exists and that the original spacetime $(M,g)$ does not contain any intertwined timelike geodesics, is as then follows: First (Section \ref{sec:setmax}), we define an ordering relation via embeddings and then essentially 'glue' together an ordered collection of regular future g-boundary extensions by taking the disjoint union and then identifying all points which are related by the ordering. This makes it straightforward to verify that the resulting object is still a regular future g-boundary extension. Since here the family we are gluing is assumed to be ordered, we can still allow $(M,g)$ to have intertwined timelike geodesics in principle (but the ordering via embeddings implicitly guarantees that these intertwined geodesics would not acquire endpoints in the considered family). This gives us maximal extensions in a set-theoretic sense by a standard Zorn's Lemma type argument, inspired by Choquet-Bruhat and Geroch's \cite{ChoquetGeroch} proof of the uniqueness of the maximal Cauchy development (see also Ringström's \cite{Ringstrom} detailed presentation of this proof), cf.~Corollary \ref{cor:existence_of_set_theoretic_max}.

\begin{theorem}\label{thm:main1_intro}
  Let $(M,g)$ be a $C^{2}$ spacetime and $\mathcal{I}$ a partially ordered set of equivalence classes of regular future g-boundary  extensions (for the definitions of the equivalence and the ordering relation see Definition \ref{def:equiv of Ns} resp.\ Definition \ref{def:order on Ns}). Then there exists a maximal element for $\mathcal{I}$, i.e.\ there exists $[N]_{\mathrm{max}}\in \Iset$ which satisfies that if $ [N]_{\mathrm{max}}\leq [N] $ for any $[N]\in \Iset$ one must already have equality $[N]_{\mathrm{max}}=[N]$
\end{theorem}

We would like to point out at this point that the extra conditions required on the topology in the definition of a {\em regular future g-}boundary extension (beyond being a topological manifold with boundary for which the interior is homeomorphic to $M$) are necessary in our proof. The rough reason for this is that these conditions fix a preferred topology on the extension based on timelike thickenings (Definition \ref{def:thickenings}) and provide a (very useful) neighborhood basis for points on the boundary. This allows us to control the topology as we pass to the quotient. 

To obtain uniqueness in Section \ref{subsec:actualmax} an extra obstruction has to be taken into account: the (possible) existence of intertwined timelike geodesics, can lead to the existence of inequivalent maximal extensions. This problem is already known from the study of the Taub-NUT spacetime (or the simpler example of Misner \cite{Misner}), which has two inequivalent maximal conformal boundary extensions (see e.g~\cite{Chrusciel}, Section 5.7 for a discussion). This leads us to the main result (cf. Theorem \ref{theo:uniquemaximalextension}) of our paper:
\begin{theorem}\label{thm:main2_intro}
    Let $(M,g)$ be a strongly causal $C^2$ spacetime. If $(M,g)$ is regular future g-boundary extendible and does not contain any intertwined future directed timelike geodesics, then there exists a unique maximal regular future g-boundary extension in the sense of Definition \ref{def:Maximal_regular_extension}. 
\end{theorem}
The proof here is rather similar to the above: We do an analogous 'take the disjoint union and then identify' quotient construction for two arbitrary regular future g-boundary extensions but now use that we excluded intertwined timelike geodesics (instead of the ordering) to show that the quotient space is again a regular future g-boundary extension. This then implies that any two set-theoretic maximal elements have to coincide.

Finally we note that while our proofs are based on Zorn's Lemma,
our second main Theorem, Theorem \ref{thm:main2_intro}, can also be obtained more constructively without invoking Zorn's Lemma, cf. the discussion in Remark \ref{rem:dezorn}.

\paragraph*{Acknowledgements}  This article originally started from work on MvdBS' Masters thesis written at the University of Tübingen. We would like to thank Carla Cederbaum for her support and bringing this collaboration together. We would further like to thank Eric Ling for bringing some of these problems to our attention and stimulating discussions. 
MG acknowledges the support of the German Research Foundation through the SPP2026 ''Geometry at Infinity'' and the excellence cluster EXC 2121 ''Quantum Universe'' -- 390833306. MvdBS thanks Carla Cederbaum for her financial support during the development of this research project, the Studienstiftung des deutschen Volkes for granting him a scholarship during his Master studies and Felix Finster for his support.

\section{Future boundary extensions}\label{sec:maximalboundary}

\par In the first place, we consider a $C^{k}$ spacetime as a connected time-oriented Lorentzian manifold $(M,g)$ without boundary with a $C^{k}$-regular metric $g$. Furthermore, timelike curves are smooth curves whose tangent vector is timelike everywhere. Note that, comparing with our main sources, this convention coincides with the one in \cite{Sbierski2022}, but differs from the one in \cite{GallowayLing}, where they use piecewise smooth timelike curves. However this does not make a difference for the resulting timelike relations.  The following basic concepts play an important role in our study.

\begin{definition}[$C^{l}$ spacetime extension]\label{def:classicalextension}
    Fix $k\geq 0$ and let $0\leq l\leq k$. Let $(M,g)$ be a $C^{k}$ spacetime with dimension $d$. A $C^{l}$\emph{ spacetime extension} of $(M,g)$ is a proper isometric embedding $\iota$ 
\begin{equation*}
        \iota\;:\; (M,g)\hookrightarrow (M_{\textrm{ext}},g_{\textrm{ext}})
    \end{equation*}
where $(M_{\textrm{ext}},g_{\textrm{ext}})$ is $C^{l}$ spacetime of dimension $d$. If such an embedding exists, then $(M, g)$ is said to be $C^{l}$ extendible. The topological boundary of $M$ within $M_{\textrm{ext}}$ is $\partial \iota(M)\subset M_{\textrm{ext}}$. By a slight abuse of notation we will sometimes also call $(M_{\textrm{ext}},g_{\textrm{ext}}$) the extension of $(M,g)$, dropping the embedding $\iota$.
\end{definition}

\begin{definition}[Future and past boundaries]
    We define the \emph{future boundary} $\partial^{+} \iota (M)$ and \emph{past boundary} $\partial^{-} \iota (M)$:
\begin{equation*}
        \partial^{+}\iota(M)\coloneqq \{p\in \partial \iota(M)\;:\: \exists\; \textrm{f.d.t.l. curve }\; \gamma : [0, 1] \rightarrow M_{\textrm{ext}} \;\textrm{with}\; \gamma(1)=p,\; \gamma([0,1))\subset \iota(M) \}
\end{equation*}
\begin{equation*}
        \partial^{-}\iota(M)\coloneqq \{p\in \partial \iota(M)\;:\: \exists\; \textrm{f.d.t.l. curve}\; \gamma : [0, 1] \rightarrow M_{\textrm{ext}} \;\textrm{with}\; \gamma(0)=p,\; \gamma((0,1])\subset \iota(M) \}
\end{equation*}
where ``f.d.t.l. curve" stands for future directed timelike curve. 
\end{definition}

Note that it does in general not hold that $\pr\iota(M)=\pr^+\iota(M)\cup \pr^-\iota(M)$ but only that $\pr^+\iota(M)\cup \pr^-\iota(M)\neq \emptyset$ (cf.~\cite{Sbierski2018}). One of the advantages of working with $\pr^+\iota(M)$ and $\pr^-\iota(M)$ is that, as we mentioned in the introduction, if one of them is empty, the other becomes particularly nice.
\begin{theorem}[Theorem 2.6 in \cite{GallowayLing}]\label{theo:pastboundary}
    Let $\iota$: $(M,g)\rightarrow (M_{\textrm{ext}},g_{\textrm{ext}})$ be a $C^{0}$-extension. If $\partial^{+}\iota(M)=\emptyset$, then $\partial^{-}\iota(M)$ is an achronal topological hypersurface.
\end{theorem}

As advertised in the introduction our main extension concept will not be the spacetime extensions of  Definition \ref{def:classicalextension} but rather certain 'future boundary extensions', a concept which we will develop now.  Of course all our constructions (with all their caveats) should work analogously for a past boundary.  

\begin{definition}[Candidate for a future boundary extension]\label{def:candidate_boundary_ext}
    Let $(M,g)$ be a spacetime with an at least $C^2$-metric and let $(N,\tau)$ be a topological space. If there exists a topological embedding $\iota : M\to N$  such that $\iota(M)$ is open and $\overline{\iota(M)}=N$, then we say that $((N,\tau),\iota)$ is a candidate for a future boundary extension of $(M,g)$. We may suppress both $\tau$ and $\iota$ notationally if they are clear from context.
\end{definition}

 We denote by $\pi^{\textrm{TM}}: TM\rightarrow M$ the natural projection map from the tangent bundle to $M$. We also fix a complete Riemannian background metric $h^{TM}$ on $TM$ and throughout this section all distances in $TM$ will be measured with respect to this background metric.\footnote{As is usually the case with these constructions, none of our arguments will require an explicit form of this background metric and, while the concrete sets $O_{X,r}$ will depend on $h^{TM}$ for the purpose of testing the topology on $N$ all choices of $h^{TM}$ are equivalent. In particular if $N$ is a future boundary extension of $M$ (cf.\ Definition \ref{def:regular-g-boundary-ext}), then whether $N$ is a \emph{regular} future g-boundary extension (cf.\ Definition \ref{def:regular-g-boundary-ext}) will not depend on this choice.} 
We denote by $T_{t}M$ the \emph{set of timelike tangent vectors}, i.e., 
    \begin{equation*}
        T_{t}M \coloneqq \{X\in TM: g(X,X)<0\}.
    \end{equation*}

Before we can proceed we need to do some preparatory work defining certain sets based around timelike geodesics of $M$ which will play an important role in describing  regularity of extensions at the boundary via topological properties. Given a fixed $X\in TM$ and $r>0$, let $B_r(X)$ denote the open ball in $TM$ around $X$. Moreover, for any $X\in TM$, let $\gamma_X: (a_X,b_X)\to M$ be the unique inextendible geodesic in $M$ with initial data $\gamma_X(0)=\pi^{TM}(X), \dot{\gamma}_X(0)=X$. Note that $X\mapsto a_X$ is upper semi-continuous and $X\mapsto b_X$ is lower semi-continuous.

\begin{definition}[Timelike thickening]\label{def:thickenings}
Let $(M,g)$ and $((N,\tau),\iota)$ as above. For $X\in T_tM$ and $r>0$ the \emph{timelike thickening} of radius $r$ generated from $X$ is
\begin{equation}
    O_{X,r}\coloneqq O_{X,r}^{\pr}\cup O_{X,r}^{\textrm{int}}
\end{equation}
where the \emph{timelike boundary thickening} $O_{X,r}^{\pr}$ and the \emph{timelike interior thickening} $O_{X,r}^{\textrm{int}}$ are defined as follows:
\begin{align}
   O_{X,r}^{\textrm{int}}\coloneqq & \{(\iota\circ \gamma_{Y})((0,b_{Y})) : Y\in B_r(X)\cap T_tM\}\label{eq:interiorthickening_new} 
\end{align}
and
\begin{align}
    O_{X,r}^{\pr}\coloneqq & \{\lim_{t\to b_Y^-}(\iota\circ\gamma_{Y})(t) : Y\in B_r(X)\cap T_tM \;\,\textrm{s.t.\;this\;limit\;exists\;in}\;N \}. \label{eq:boundarythickening_new}
\end{align}
\end{definition} 

These are natural analogues of the thickenings of null geodesics considered in \cite{Chrusciel}.

\begin{remark}\label{rem:tlthickenings}
Note that $O_{X,r}$, while indexed by objects intrinsic to $(M,g)$, also depends on $(N,\tau)$ and the embedding $\iota : M\to N$. In all our applications $(M,g)$ will be fixed, however, we will sometimes need to consider different $N$. Whenever there is any chance of confusion we will indicate in which $N$ we are considering the timelike thickening by writing $O^{N}_{X,r}$ instead of merely $O_{X,r}$. 
\end{remark}

Now we are ready to define our concept of (regular) future (g-)boundary extensions:

\begin{definition}[Regular future g-boundary extension]\label{def:regular-g-boundary-ext}
   Let $(M,g)$ be a $C^2$-spacetime. We say that a topological manifold with boundary $N$ is a \emph{future boundary extension} of $(M,g)$ if there exists a homeomorphism
   $$\iota: M\to \mathrm{int}(N)$$ 
   and for any $p\in \pr N$ there exists a future directed timelike curve $\gamma:[0,1)\to M$ with $p=\lim_{t\to 1^-} \iota(\gamma(t))$. If further
   \begin{enumerate}
       \item for any $p\in \pr N$ there exists a future directed timelike \emph{geodesic} $\gamma:[0,1)\to M$ with $p=\lim_{t\to 1^-} \iota(\gamma(t))$
       \item and all timelike thickenings $O^N_{X,r}$ are open and for any $p\in \pr N$ and any future directed timelike geodesic $\gamma:[0,1)\to M$ with $p=\lim_{t\to 1^-} \iota(\gamma(t))$ the collection $\{O^N_{\dot{\gamma}(1-\frac{1}{n}),\frac{1}{m}}: n,m\in \mathbb{N}\}$ is a neighborhood basis of $p$, 
 \end{enumerate} 
then we say that $N$ is a \emph{regular future g-boundary extension}.
\end{definition}

Let us first note that in Section \ref{sec:compatibility} we show that for globally hyperbolic $(M,g)$ any $C^{0}$-spacetime extension $(\mext,\gext)$ in the sense of Definition \ref{def:classicalextension} with empty past boundary gives rise to a future boundary extension $N:=\pr^+\iota(M)\cup \iota(M) \subset \mext$. If $(\mext,\gext)$ is a $C^2$-extension with empty past boundary, then $N$ will be a regular future g-boundary extension. This suggests viewing conditions (1) and (2) in Definition \ref{def:regular-g-boundary-ext} as hidden regularity assumptions and is the reason we introduced the name of \emph{regular} future g-boundary extensions. The ''g'' refers to ''geodesic'' as we demand that all points in the boundary are reached by timelike geodesics and also refers back to old constructions of a ''geodesic boundary'' by Geroch and others, see \cite{Gerochgboundary} and \cite{Geroch2}, highlighting some similarities in spirit to our approach. 
 The idea of Geroch's g-boundary is the following: given a geodesically incomplete spacetime $M$ one considers the set of incomplete geodesics. This set can be endowed with an equivalence relation which, intuitively, considers as equivalent incomplete geodesics that become arbitrarily close (as they approach the singularities of $M$). This set of equivalence classes is called the g-boundary. Note that the resulting object of attaching this g-boundary to the original spacetime $M$ is only a topological space: i.e. in general it is not a manifold anymore and issues with non-Hausdorffness may appear. However, it was more recently shown that it is possible to find a finer topology on the topological space that arises from 'attaching' the g-boundary to the original spacetime $M$ such that this space becomes Hausdorff in the new topology (\cite{FloresHerreraSanchez}). It remains to be seen whether this could be used in actually constructing regular future g-boundaries or proving regular future g-boundary extendibility.

\begin{remark}\label{rem:whyboundary}
Our main reason for switching to work with topological manifolds with boundary instead of the classical concept of spacetime extensions from Definition \ref{def:classicalextension}, where the extension is itself again a spacetime without boundary, is that a uniqueness result for a maximal extension (with the ''standard'' ordering defined via the existence of a \emph{global} embedding) is clearly impossible when going beyond the boundary as one can freely modify the topology of $\mext\setminus \overline{\iota(M)}$ as well as the extended metric $\gext$ on $\mext\setminus \overline{\iota(M)}$. However, recent results of (Sbierski, \cite{Sbierski2022}) show that there is a strong local uniqueness up to and including the boundary. We tried adapting the definition of an ordering relation to only demand the existence of an embedding of some open neighborhood of the boundary (cf.\ Remark \ref{rem:relation_w_extensions_in_marco_thesis}), however for such modified orderings it is not readily apparent that set theoretic maximal elements even have to exist: The problem here appears to be that when trying to construct set theoretic upper bounds via taking unions over the elements in an infinite totally ordered set of extensions (and identifying appropriately) one quickly runs into the issue that -- in order to ensure that the resulting object is a manifold -- we would need a common neighborhood of the boundary into which all other neighborhoods progressively embed, however such a  common neighborhood need not exist, as the considered neighborhoods could contract to just the boundary itself. Indeed we expect that this process would generally only produce a manifold with boundary. Working with topological manifolds with boundary from the beginning avoids these issues.
\end{remark}

\subsection{Preliminary topological considerations}\label{sec:topologicalremarks}

As we already remarked in the introduction, condition (2) in Definition \ref{def:regular-g-boundary-ext} will be necessary to control the topology of our upcoming quotient space constructions. In this preliminary section we will give a first example on how (2) controls the topology by showing that  it guarantees second countability, even if $(N,\tau)$ is not assumed to be a manifold with boundary already.

For this we now define timelike thickenings in $M$ itself (in analogy of timelike thickenings in candidates $((N,\tau),\iota )$ for  future boundary extensions)  by
\begin{align}
   O_{X,r}^{M}\coloneqq & \{\gamma_{Y}((0,b_{Y})) : Y\in B_r(X)\cap T_tM\}\label{eq:Mthickening}
\end{align} 
for $X\in T_tM$ and $r>0$. 
 We are interested in the interplay between the topologies of $M$ and $N$ and properties of the sets $O^M_{X,r}$ and $O^N_{X,r}$. 

\begin{remark}\label{rem:OM_def_and_open} 
Clearly for any candidate for a future boundary extension $((N,\tau),\iota )$ of $(M,g)$ we have $O_{X,r}^M=\iota^{-1}( O_{X,r}^{\textrm{int}})=\iota^{-1}(O^N_{X,r}\cap\iota(M))$. Further, $O_{X,r}^M$ is open in $M$: First, $\{s\cdot Y:Y\in B_r(X)\cap T_tM, 0<s<b_Y \}=\{s\cdot Y:Y\in B_r(X), 0<s<b_Y \}\cap T_tM \subset TM$ is open  by lower semi-continuity of $Y\mapsto b_Y$, second the exponential map $\exp :\mathcal{D}\subset TM \to M$ mapping $X$ to $\gamma_X(1)$ is an open map and lastly $\{s\cdot Y:Y\in B_r(X)\cap T_tM, 0<s<b_Y \} \subset \mathcal{D}$ and $O_{X,r}^M=\exp(\{s\cdot Y:Y\in B_r(X)\cap T_tM, 0<s<b_Y \})$. 
\end{remark}

So there is, as expected, a quite strong relationship between $O^M_{X,r}$ and $O^N_{X,r}$. On the other hand, the $O^N_{X,r}$ are a priori relatively independent of the topology on $N$ (except for $O^N_{X,r}\cap \iota(M)$ having to be open) and demanding ''regularity'' is exactly forcing a stronger relation between the $O^N_{X,r}$ and the topology on $N$. We define

\begin{definition}[Candidate for a regular future g-boundary extension]\label{def:candidate_regular_boundary_ext}
    Let $(M,g)$ be a spacetime with $C^2$-metric. We say that a candidate $((N,\tau),\iota )$ for a future boundary extension is a candidate for a regular future g-boundary extension if all timelike thickenings $O^N_{X,r}$ are open and for any $p\in N\setminus \iota(M)$ there exists a future directed timelike geodesic $\gamma:[0,1)\to M$ with $p=\lim_{t\to 1^-} \iota(\gamma(t))$ and for any such geodesic $\gamma$ the collection $\{O^{N}_{\dot{\gamma}(1-\frac{1}{n}),\frac{1}{m}}: n,m\in \mathbb{N}\}$ is a neighborhood basis for $p$.
\end{definition}

We will next prove that \emph{if} $N$ is a candidate for a regular future g-boundary extension of $M$, then the topology on $N$ is always second countable and can be described entirely by the family of timelike thickenings in $N$ and the topology on $M$.

\begin{lemma}\label{lem:AA2} Let $(M,g)$ be a spacetime with an at least $C^2$-metric and let $((N,\tau),\iota )$ be a candidate for a regular future g-boundary extension of $(M,g)$. Then for any countable dense subset  $\{X_i\}_{i\in \mathbb{N}}$ of $T_tM$ and any countable basis $\{U_i\}_{i\in \mathbb{N}}$  for the manifold topology of $M$ the collection
$$\mathcal{B}_{t}:=\{O^{N}_{X,r}: X\in\{X_i\}_{i\in \mathbb{N}}\;\mathrm{and}\;0<r\in \mathbb{Q}\}\cup \{\iota(U_i)\}_{i\in \mathbb{N}}$$
is a countable basis for $\tau$.
\end{lemma}

\begin{proof}

We need to show that for each $\tau$-open $U\subset N$ and every $p\in U$ there exists $O^{N}_{X,r}\in \mathcal{B}_t$ with $p\in O^{N}_{X,r}$ and $O^{N}_{X,r}\subset U$. If $p\in U\cap\iota(M)$ this immediately follows from $\iota(M)$ being open, $\iota$ being an embedding and $\{\iota(U_i)\}_{i\in \mathbb{N}}$ being a basis for the topology on $M$. So assume $p\in N\setminus \iota(M)$. Since by assumption $\{O^{N}_{\dot{\gamma}(1-\frac{1}{n}),\frac{1}{m}}: n,m\in \mathbb{N}\}$ is then a neighborhood basis for $p$, there exist $n,m$ such that $O^{N}_{\dot{\gamma}(1-\frac{1}{n}),\frac{1}{m}}\subset U$. This is almost what we need except that $\dot{\gamma}(1-\frac{1}{n})$ might not belong to the collection $\{X_i\}_{i\in \mathbb{N}}$. By density of $\{X_i\}_{i\in \mathbb{N}}$ there exists $i\in \mathbb{N}$ s.t.\ $\dot{\gamma}(1-\frac{1}{n}) \in B_{\frac{1}{2m}}(X_i)$. Then by the triangle inequality $B_{\frac{1}{2m}}(X_i)\subset B_{\frac{1}{m}}(\dot{\gamma}(1-\frac{1}{n}))$  and hence $O^{N}_{X_i,\frac{1}{2m}}\in \mathcal{B}_t$ satisfies $p\in O^{N}_{X_i,\frac{1}{2m}}$ and $O^{N}_{X_i,\frac{1}{2m}}\subset O^{N}_{\dot{\gamma}(1-\frac{1}{n}),\frac{1}{m}} \subset U$.
\end{proof}

Hence establishing that a candidate for a regular future g-boundary extension is indeed  a regular future g-boundary extension boils down to finding homeomorphisms from open neighborhoods of  ''boundary points'' $p\in N\setminus \iota(M)$ to open subsets in the half space $[0,\infty)\times \mathbb{R}^{d-1}$ (clearly, $\iota$ induces a manifold structure on the ''interior'' $\iota(M)$ and $\iota$ being a homeomorphism between $M$ and the open set $\iota(M)\subset N$ takes care of compatibility of charts) and showing Hausdorffness while second countability then follows automatically.

\section{Compatibility with other extension concepts}\label{sec:compatibility}
As a further preliminary step, let us -- as promised -- investigate under which conditions we can strip down a spacetime extension $\iota: M \to \mext$ in the sense of Definition \ref{def:classicalextension} to just $\iota(M)\cup \pr^+\iota(M)$ while retaining a sensible structure, namely that of a topological manifold with boundary or even of a regular future g-boundary extension, on the resulting space.

First, we discuss under which sufficient conditions,  given a (low-regularity) extension $\iota : M\rightarrow \mext$, the subspace $\iota(M)\cup \pr^{+}\iota(M)$ is a topological manifold with boundary. If we endow $\iota(M)\cup \pr^{+}\iota(M)$ with the subspace topology, it directly follows that it is Hausdorff and second countable (inherited properties from the manifold topology in $\mext$). However, it does not hold, in general, that $\iota(M)\cup \pr^{+}\iota(M)$ is a topological manifold with boundary. In particular, it is not clear under which conditions on $M$ and on the extension $\iota$, for points in $\pr^{+}\iota(M)$ there exists an open neighborhood $V$  homeomorphic to a relatively open subset of $[0,\infty)\times \mathbb{R}^{d-1}$. The following result in \cite{Sbierski2018} plays an important role in investigating this.
 
\begin{proposition}[Proposition 1 in \cite{Sbierski2018}]
\label{Sbierskifuturechart}
Let $\iota : M\rightarrow \mext$ be a $C^{0}$-extension of a globally hyperbolic Lorentzian manifold $(M,g)$ and let $p\in\pr^{+}\iota(M)$. For every $\delta> 0$ there exists a chart $\varphi : V \rightarrow \rec\coloneqq (-\varepsilon_{0},\varepsilon_{0})\times (-\varepsilon_{1},\varepsilon_{1})^{d-1} $ with $\varepsilon_{0},\varepsilon_{1}>0$ with the following properties:
\begin{enumerate}
    \item $p\in V$ and $\varphi(p)= (0, . . . , 0)$.
    \item $|g_{\mu\nu} - \eta_{\mu\nu}| < \delta$, where $\eta_{\mu\nu}$ is the Minkowski metric.
    \item There exists a Lipschitz continuous function $f : (-\varepsilon_{1}, \varepsilon_{1})\rightarrow (-\varepsilon_{0},\varepsilon_{0})$ with the following
properties:
\begin{equation}
    F_{<}\coloneqq \{(x_{0},x)\in \rec | x_{0}<f(x)\}\subset \varphi(\iota(M)\cap V)
\end{equation}
\begin{equation}
    F_{=}\coloneqq \{(x_{0},x)\in \rec | x_{0}=f(x)\}\subset \varphi(\pr^{+}\iota(M)\cap V)
\end{equation}
Moreover, $F_{=}$ is achronal in $\rec$ and $\varphi$ is called a future boundary chart.
\end{enumerate}
\end{proposition}
\par The previous Proposition implies that points beneath the graph of the Lipschitz function $f$ are in the inside of the ``original" spacetime $\iota(M)$. An easy way to ensure that points above the graph of $f$ are in $\mext\setminus\overline{\iota(M)}$ (as, in general, it cannot be ruled out that some of these points are in $\iota(M)$ or $\pr\iota(M)$, cf.\ the comments in \cite{Sbierski2022}) is to assume that the past boundary is empty. Under this assumption, we immediately have the following:
\begin{lemma}\label{lem:no_further_intersections}
  Let $\iota : M\rightarrow \mext$ be a $C^{0}$-extension of a globally hyperbolic Lorentzian manifold $(M,g)$ such that $\pr^-\iota(M)=\emptyset$. 
    Then for any smooth future directed timelike curve $\gamma: [0,b)\to \mext$ in $\mext$ with
    $\gamma(0)\in \iota(M)$ and $\lim_{s\to b_0^-} \gamma(s) \in \pr^+\iota(M)$ for some $b_0\in (0,b)$ there exists a unique $s\in (0,b)$ such that $\gamma([0,s))\subset \iota(M)$, $\gamma(s)\in \pr^+\iota(M) $, $\gamma((s,b))\subset  \mext\setminus (\iota(M)\cup \pr^{+}\iota(M)) $  and $s=b_0$.  
\end{lemma}

\begin{proof}
    Let $\gamma: [0,b)\to \mext $ be a suitable timelike curve. 
    The set $ \{t: \gamma(t)\in \pr^+\iota(M)\}$ is non-empty by assumption and we set $s:=\inf \{t: \gamma(t)\in \pr^+\iota(M)\}$. By openness of $\iota(M)$ and continuity of $\gamma$ we have $s>0$ and $\gamma([0,s))\subset \iota(M)$. Clearly $s\leq b_0$ by definition, so $s\in (0,b)$. It remains to show $\gamma((s,b))\subset  \mext\setminus (\iota(M) \cup \pr^+\iota(M)) $. 
    Assume that $\gamma(b')\in \iota(M) \cup \pr^+\iota(M) $ for some $s<b'< b$. Achronality of $\pr^+\iota(M)$, which follows from (the time reversed version of) Theorem \ref{theo:pastboundary}, implies that $\gamma(b')\notin \pr^+\iota(M)$, hence  $\gamma(b')\in \iota(M) $. We now proceed as before: setting $s':=\inf \{t\in (s,b'):  \gamma(t)\in \iota(M) \}$ we have $\gamma((s',b']) \subset \iota(M)$ and $\gamma(s')\in \pr\iota(M)$. This contradicts $\pr^-\iota(M)$ being empty.\end{proof}

Since the first part of the lemma in particular applies to vertical coordinate lines in the chart $\varphi$, it is clear that points below the graph of $f$ are inside $\iota(M)$ while points above the graph of $f$ are outside of $\iota(M) \cup \pr^+\iota(M)$. 
 Therefore, given a globally hyperbolic spacetime and considering low regularity extensions with a disjoint future and past boundary, it follows that $\iota(M)\cup \pr^{+}\iota(M)$ is a topological manifold with boundary: taking around every point $p\in \pr^{+}\iota(M)$ a future boundary chart $\varphi: V\rightarrow \rec$ and defining the homeomorphism $\phi: \rec\rightarrow \mathbb{R}^{d},\;(x_{0},x)\mapsto(x_{0}-f(x),x)$, it follows that every $V\cap(\iota(M)\cup\pr^{+}\iota(M))$ is locally homeomorphic to a relatively open subset of $[0,\infty)\times \mathbb{R}^{d-1}$ (with homeomorphism $\Tilde{\phi}\coloneqq\phi\circ\varphi : V\subset M\rightarrow \Tilde{\phi}(V)\subset [0,\infty)\times \mathbb{R}^{d-1}$). Moreover, the fact that $f$ is only a Lipschitz function (so $\Tilde{\phi}=\phi\circ\varphi$ is Lipschitz but not smooth) is why, in general, $\iota(M)\cup \pr^{+}\iota(M)$ is a topological manifold but not a smooth manifold\footnote{Since $f$ is Lipschitz we could probably have worked with Lipschitz manifolds with boundary throughout (i.e., from Definition \ref{def:regular-g-boundary-ext}), but we didn't see an immediate way to take advantage of the additional Lipschitz structure, so we stuck with topological manifolds with boundary.}. We have thus shown:

\begin{lemma} Let $(M,g)$ be globally hyperbolic and $(\mext,\gext)$ a $C^0$ 
extension with empty past boundary, then $N:=\iota(M)\cup \pr^{+}\iota(M)$ with the subspace topology induced from $\mext$ is a topological manifold with boundary and a future boundary extension of $(M,g)$.    
\end{lemma}

\par The following Lemma establishes that, given a regular enough extension of a globally hyperbolic spacetime, every point of its future boundary is intersected by a timelike geodesic.

\begin{lemma}[Lemma 3.1 in \cite{Sbierski2022}]\label{lem:geodesic_to_boundary}
Let $(M, g)$ be a $C^{2}$ time-oriented and globally hyperbolic Lorentzian manifold and let $(\mext,\gext)$, $\iota : M \rightarrow \mext$  be a $C^2$-extension. 
Let 
 $\gamma:[-1,0)\rightarrow M$ be a future directed and future inextendible
causal $C^{1}$-curve such that $\lim_{s\to0}(\iota\circ\gamma)(s)\eqqcolon p\in \pr\iota(M)$ exists. Then there is a smooth timelike geodesic
$\sigma : [-1,0]\rightarrow \mext$ with $\sigma|_{[-1,0)}$ mapping into $\iota(M)$ as a future directed timelike geodesic and $\sigma(0)=p$. In
particular $p\in\pr^{+}\iota(M)$ and there exists a boundary chart such that $\iota\circ\gamma$ is ultimately
contained in $F_{<}\coloneqq \{(x_{0},x)\in \rec | x_{0}<f(x)\}$.
\end{lemma}

To establish that given a regular enough extension of a globally hyperbolic spacetime for which the past boundary is empty $N:=\iota(M)\cup \pr^{+}\iota(M)$ is a regular g-boundary extension it only remains to show that all sets $O^N_{X,r}$ are open (in the subspace topology $\tau_s$ induced on $N$ from $N \subset \mext$) and  for any $p\in \pr^{+}\iota(M)$ and any future directed timelike geodesic $\gamma:[0,1)\to M$ with $p=\lim_{t\to 1^-} \iota(\gamma(t))$ the collection $\{O^N_{\dot{\gamma}(1-\frac{1}{n}),\frac{1}{m}}: n,m\in \mathbb{N}\}$ is a neighborhood basis of $p$ (for $\tau_s$). 

\begin{lemma} \label{lem:compatibility:regularityOK} Let $(M,g)$ be a globally hyperbolic $C^2$ spacetime and $(\mext,\gext)$ a $C^2$ 
spacetime extension with empty past boundary. Set $N:=\iota(M)\cup \pr^{+}\iota(M) \subset \mext$. Let $p\in \pr^{+}\iota(M)$ and let $\gamma:[0,1)\to M$ be a future directed timelike geodesic with $p=\lim_{t\to 1^-} \iota(\gamma(t))$. The collection $\{O^N_{\dot{\gamma}(1-\frac{1}{n}),\frac{1}{m}}: n,m\in \mathbb{N}\}$ is a $\tau_s$-neighborhood basis of $p$. Further, any $O^N_{X,r}$ is $\tau_s$-open.
\end{lemma}

\begin{proof}
  We first show that any $O^N_{X,r}$ is $\tau_s$-open.
 Take an arbitrary $X\in T_tM, r>0$. Then we define 
    $$O^{\textrm{ext}}_{X,r}\coloneqq \{\gamma^{\textrm{ext}}_{Y}((0,b^{\textrm{ext}}_{Y})) : Y\in B_{r}(X)\cap T_{t}M\}$$
    where $\gamma^{\textrm{ext}}_{Y} : (0,b^{\textrm{ext}}_{Y})\rightarrow \mext$ is the unique future inextendible timelike geodesic in $\mext$ with initial data $\dot{\gamma}^{\textrm{ext}}(0)=(T\iota)(Y)\in T\mext \cap T_t\mext$. Note that $b_Y\leq b^{\textrm{ext}}_{Y} $ and $\gamma^{\textrm{ext}}_{Y}=\iota\circ \gamma_Y$ on $(0,b_Y)$. If $b_Y< b^{\textrm{ext}}_{Y}$, then by Lemma \ref{lem:no_further_intersections}, $\gamma^{\textrm{ext}}_{Y}(b_Y)\in \pr^+\iota(M)=N\setminus\iota(M)$ and $\gamma^{\textrm{ext}}_{Y}(t)\notin \iota(M)\cup \pr^{+}\iota(M)=N$ for any $t\in (b_Y, b^{\textrm{ext}}_{Y})$. Therefore, $$O^N_{X,r}=O^{\textrm{ext}}_{X,r}\cap N.$$ This, together with openness of $O^{\textrm{ext}}_{X,r}$ in $\mext$ (cf. Remark \ref{rem:OM_def_and_open} noting that\footnote{Assuming the Riemannian background metric $h^{\textrm{ext}}$ on $T\mext$ is chosen to satisfy $B^{h^{\textrm{ext}}}_r(X')=(T\iota)(B^h_r(X))$, but for given $X,r$ this can always be achieved. Else one could also use different radii to obtain appropriate subset relations.}
 $O^{\textrm{ext}}_{X,r}$ as defined above of course equals $O^{\mext}_{X',r}$ as defined in \eqref{eq:Mthickening} for $X':=(T\iota)(X)\in T_t\mext$), implies that $O^N_{X,r}$ is open in the subspace topology on $N$.

 To show that the collection $\{O^N_{\dot{\gamma}(1-\frac{1}{n}),\frac{1}{m}}: n,m\in \mathbb{N}\}$ is a $\tau_s$-neighborhood basis of $p$  note first that by the above we also have $O^N_{\dot{\gamma}(1-\frac{1}{n}),\frac{1}{m}}=O^{\textrm{ext},\varepsilon}_{\dot{\gamma}(1-\frac{1}{n}),\frac{1}{m}}\cap N$ for any $\varepsilon>0$, where 
 $$O^{\textrm{ext},\varepsilon}_{X,r}:=\{\gamma^{\textrm{ext}}_{Y}((0,\min(b_Y+\varepsilon,b^{\textrm{ext}}_{Y}))) : Y\in B_{r}(X)\cap T_{t}M\}$$ 
 So the problem reduces to arguing that the sets $O^{\textrm{ext},\varepsilon}_{\dot{\gamma}(1-\frac{1}{n}),\frac{1}{m}}$ with $n,m\in \mathbb{N},\varepsilon >0$ form a neighborhood basis for $p=\gamma^{\textrm{ext}}(1)$ in $\mext$, that is for any open set $U$ around $p$ in $\mext$ there exist $n,m\in \mathbb{N}$ and $\varepsilon>0$ such that  $O^{\textrm{ext},\varepsilon}_{\dot{\gamma}(1-\frac{1}{n}),\frac{1}{m}}$ is open and $O^{\textrm{ext},\varepsilon}_{\dot{\gamma}(1-\frac{1}{n}),\frac{1}{m}}\subset U$. To see this, first fix $n$ such that $\gamma^{\textrm{ext}}([1-\tfrac{1}{n}, 0])\subset U$. Then set $X:=\dot{\gamma}(1-\tfrac{1}{n})$ and choose $\varepsilon>0$ such that $b_X+\varepsilon<b_X^{\textrm{ext}}$ and $\gamma_X^{\textrm{ext}}([0,b_X+\varepsilon])\subset U$. Finally, by continuous dependence of $\mext$-geodesics on their initial data, there exists a neighborhood $V$ of $(T\iota)(X)$ in $T\mext$ such that $V\subset (T\iota)(TM)$, $b_Y^{\textrm{ext}}>b_Y+\varepsilon$ for all $Y\in T\iota^{-1}(V)$ and $\gamma^{\textrm{ext}}_Y([0,b_Y+\varepsilon])\subset U$ for all $Y\in T\iota^{-1}(V)$. Now we just need to choose $m$ with $B_{\frac{1}{m}}(X)\subset T\iota^{-1}(V)$ and see that $O^{\textrm{ext},\varepsilon}_{X,\frac{1}{m}}\subset U$ is the desired neighborhood. \end{proof}

Collecting results we have shown

\begin{proposition}\label{prop:compatibility}
    Let $(M,g)$ be a globally hyperbolic $C^2$ spacetime and $(\mext,\gext)$ a $C^2$ 
spacetime extension with empty past boundary, then $N:=\iota(M)\cup \pr^{+}\iota(M)$ with the subspace topology induced from $\mext$ is a topological manifold with boundary and a regular future g-boundary extension of $(M,g)$.
\end{proposition}

\section{Ordering relation and existence of maximal elements}\label{sec:setmax}

\subsection{Partial ordering and equivalence of regular future g-boundary extensions}
In this short section we introduce an equivalence relation on the collection of regular future g-boundary extensions.

\begin{definition}\label{def:equiv of Ns} 
    Let $(M,g)$ be a $C^{2}$ spacetime and $(N_1,\iota_1),(N_2,\iota_2)$ be two regular future g-boundary extensions of $M$. We say $(N_1,\iota_1)\cong (N_2,\iota_2)$ if there exists a homeomorphism (of topological manifolds with boundary) $\psi_{12}:N_1\to N_2$ that is compatible with the homeomorphisms $\iota_1:M\to \mathrm{int}(N_1)$ and $\iota_2:M\to \mathrm{int}(N_2)$, i.e., such that 
    $$\iota_2^{-1}\circ \psi_{12}\circ \iota_1 : M\to M$$ is the identity map for $(M,g)$. In other words, we demand that $\iota_2\circ \iota_1^{-1}:\iota_1(M) \to \iota_2(M)$ extends to a homeomorphism $\psi_{12}:N_1\to N_2$.
\end{definition}

Clearly this is reflexive, symmetric and transitive, so this relation defines an equivalence relation.  We denote the equivalence classes with $[(N,\iota)]$ and define the set of all equivalence classes as
\begin{equation}\label{eq:defIset}
    \Iset:=\{[(N,\iota)]: (N,\iota)\; \mathrm{is\,a\,regular\,future\,g-boundary\,extension}\}.
\end{equation} 

\begin{remark}
\label{rem:Iset_is_set}
Let us briefly justify why $\Iset$ is small enough to be a set. While the class of all $n$-dimensional topological manifolds with boundary is a proper class, the set of all $n$-dimensional topological manifolds with boundary \emph{up to homeomorphism} is indeed a set as any topological manifold with boundary can be embedded into $\R^m$ for $m$ sufficiently large. 
While we don't quite identify up to homeomorphism, i.e., given $(N_1,\iota_1)$ and $(N_2,\iota_2)$ just $N_1\cong_{\mathrm{hom}} N_2$ is insufficient to ensure $[(N_1,\iota_1)]=[(N_2,\iota_2)]$ as also the embeddings $\iota_1,\iota_2$ have to be compatible, the embeddings themselves "fix" the remaining freedom in the structure of $N$ in relation to the given fixed $M$ (note that $M$ is fixed as a set and not just up to diffeomorphism). More precisely, considering the set $\mathcal{E}$ given as
$$\{(E,i): E \subset \R^m, i\subset M\times \R^m  \;\mathrm{for\;an}\;m\in \mathbb{N} \mathrm{\;with\;}i=\mathrm{graph}(f_i)\mathrm{\;for\;a\;function\;} f_i:M\to   E   \}/\sim_h, $$
where $(E,i)\sim_h (\tilde{E},\tilde{i})$ if and only if there exists a homeomorphism $e: E\to \tilde{E} $ (where $E$ and $\tilde{E}$ are understood to be carrying the trace topology) such that $e\circ f_i=f_{\tilde{i}}$, it is readily apparent that $[(N,\iota)]\mapsto [(\Phi(N),\mathrm{graph}(\Phi\circ \iota ) )]_h$  for any embedding $\Phi$ of $N$ into $\R^m$ provides an injection from $\Iset$ into the set $\mathcal{E}$: This map is independent of the choice of embedding $\Phi$ and of the choice of representative $(N,\iota)$ of $[(N,\iota)]$. Injectivity is also easily checked: If $[(\Phi_1(N_1),\mathrm{graph}(\Phi_1\circ \iota_1 ) )]_h=[(\Phi_2(N_2),\mathrm{graph}(\Phi_2\circ \iota_2 ) )]_h$, then there exists a homeomorphism $e: \Phi_1(N_1)\to \Phi_2(N_2)$ with $e\circ \Phi_1\circ \iota_1= \Phi_2\circ \iota_2$, so $\psi_{12}:=\Phi_2^{-1}\circ e\circ  \Phi_1$ is a homeomorphism between $N_1$ and $N_2$ satisfying $\psi_{12}|_{\iota_1(M)}=\iota_2\circ \iota_1^{-1}$ and hence $(N_1,\iota_1)\cong (N_2,\iota_2)$. 

\end{remark}

 To equip $\Iset$ with a partial order we define
\begin{definition}\label{def:order on Ns} 
Let $(N_1,\iota_1),(N_2,\iota_2)$ be regular future g-boundary extensions.	We say $(N_1,\iota_1) \lesssim (N_2,\iota_2)$ if there exists an embedding (of topological manifolds with boundary) $\psi_{12}:N_1\to N_2$ compatible with $\iota_1,\iota_2$, i.e., such that 
	$$\iota_2^{-1}\circ \psi_{12}\circ \iota_1 : M\to M$$ is the identity map for $(M,g)$. In other words, we demand that $\iota_2\circ \iota_1^{-1}:\iota_1(M) \to \iota_2(M)$ extends to an embedding $\psi_{12}:N_1\to N_2$.
	
For two equivalence classes $[(N_1,\iota_1)]$ and $[(N_2,\iota_2)]$ we say $[(N_1,\iota_1)]\leq[(N_2,\iota_2)]$ if there exist  representatives $(N_1,\iota_1)$ and $(N_2,\iota_2)$ of $[(N_1,\iota_1)]$ resp.  $[(N_2,\iota_2)]$ such that $(N_1,\iota_1) \lesssim (N_2,\iota_2)$. 
\end{definition}

Note that $[(N_1,\iota_1)]\leq [(N_2,\iota_2)]$ if and only if $(N_1,\iota_1) \lesssim (N_2,\iota_2)$ for all representatives $(N_1,\iota_1)$ and $(N_2,\iota_2)$ of $[(N_1,\iota_1)]$ resp. $[(N_2,\iota_2)]$: 
Let $(N_1,\iota_1)$ and $(N_2,\iota_2)$ be representatives of $[(N_1,\iota_1)]$ and $[(N_2,\iota_2)]$ respectively such that $(N_1,\iota_1) \lesssim (N_2,\iota_2)$. We show that this implies that for any other representatives $(N'_1,\iota'_1)$ of $[(N_1,\iota_1)]$ and $(N'_2,\iota'_2)$ of $[(N_2,\iota_2)]$ it holds that $(N'_1,\iota'_1) \lesssim (N'_2,\iota'_2)$. Since $(N_1,\iota_1)\cong (N'_1,\iota'_1)$ and $(N_2,\iota_2)\cong (N'_2,\iota'_2)$ there exists a homeomorphism $\psi_{1'1}: N'_{1}\rightarrow N_{1}$ compatible with $\iota_{1}$ and $\iota'_{1}$ and a homeomorphism $\psi_{22'}: N_{2}\rightarrow N'_{2}$ compatible with $\iota_{2}$ and $\iota'_{2}$. Furthermore, as $(N_1,\iota_1) \lesssim (N_2,\iota_2)$, there exists an embedding $\psi_{12}: N_{1}\rightarrow N_{2}$ compatible with $\iota_{1}$ and $\iota_{2}$. We define the map $\psi_{1'2'}\coloneqq \psi_{22'}\circ\psi_{12}\circ\psi_{1'1} : N'_{1}\rightarrow N'_{2}$, which, by construction is an embedding (it is the composition of embeddings). It is clearly compatible with $\iota'_{1}$ and $\iota'_{2}$, which can be easily verified using that $\psi_{22'}|_{\mathrm{int}(N_2)}=\iota'_{2}\circ\iota_{2}^{-1}$, $\psi_{12}|_{\mathrm{int}(N_1)}=\iota_{2}\circ\iota_{1}^{-1}$ and $\psi_{1'1}|_{\mathrm{int}(N'_1)}=\iota_{1}\circ(\iota'_{1})^{-1}$. Hence, $(N'_1,\iota'_1) \lesssim (N'_2,\iota'_2)$. As the representatives $(N'_1,\iota'_1)$ and $(N'_2,\iota'_2)$ were chosen arbitrarily, we can conclude that $(N_1,\iota_1) \lesssim (N_2,\iota_2)$ for all representatives $(N_1,\iota_1)$ and $(N_2,\iota_2)$ of $[(N_1,\iota_1)]$ and $[(N_2,\iota_2)]$ respectively.

Hence $\leq $ indeed defines a partial ordering on the set $\Iset$ of all equivalence classes
of regular future g-boundary extensions.

\begin{remark}\label{rem:relation_w_extensions_in_marco_thesis} For spacetime extensions $\mext$ we had considered the following definition in the second author's Master thesis (see \cite{vdbs}). Let $\iota_{1} : M\rightarrow \Mextone$ and $\iota_{2} : M\rightarrow \Mexttwo$ be spacetime extensions (i.e. such as in Definition \ref{def:classicalextension}) of $(M,g)$. In this Remark, no assumption on the regularity class nor on the causal properties (e.g. global hyperbolicity) of $(M,g)$ or the extensions is made. We define the following relations: 
\begin{itemize}
    \item We say that $(\Mextone,\iota_{1})=_{\pr}(\Mexttwo,\iota_{2})$ provided there exist open neighborhoods $U_{1}$ and $U_{2}$ satisfying that $\partial \iota_{1}(M)\subset U_{1}$, $\partial \iota_{2}(M)\subset U_{2}$ and $\iota_{1}^{-1}(\iota_{1}(M)\cap U_{1})\subset \iota_{2}^{-1}(\iota_{2}(M)\cap U_{2})$, and an embedding $\psi_{12}: U_{1}\rightarrow U_{2}$ whose restriction $\psi_{12} : \overline{\iota_{1}(M)}\cap U_{1}\rightarrow \overline{\iota_{2}(M)}\cap U_{2}$ is surjective (and thus a homeomorphism) and which is compatible with the extensions, i.e. such that 
    \begin{equation*}
    \label{eq:Max1}
        \iota_{2}^{-1}\circ \psi_{12}\circ \iota_{1} : \iota_{1}^{-1}(\iota_{1}(M)\cap U) \rightarrow \iota_{1}^{-1}(\iota_{1}(M)\cap U)
    \end{equation*}
    is the identity map in $\iota_{1}^{-1}(\iota_{1}(M)\cap U_{1})\subset M$. In \cite{vdbs} the second author showed (Lemma 60 in \cite{vdbs}) that this defines an equivalence relation. We label the family of equivalence classes by $\mathcal{I}_{\textrm{ext}}$.
    \item We say that $(\Mextone,\iota_{1})\leq_{\pr}(\Mexttwo,\iota_{2})$ provided there exist open neighborhoods $U_{1}, U_{2}$ satisfying that $\partial \iota_{1}(M)\subset U_{1} $, $\partial \iota_{2}(M)\subset U_{2}$ and $\iota_{1}^{-1}(\iota_{1}(M)\cap U_{1})\subset \iota_{2}^{-1}(\iota_{2}(M)\cap U_{2})$, and an embedding $\psi_{12} : U_{1}\rightarrow U_{2}$ with $\psi_{12}(\partial\iota_{1}(M))\subset\partial\iota_{2}(M)$ and which is compatible with the extensions , i.e. such that 
    \begin{equation*}
        \iota_{2}^{-1}\circ\psi_{12}\circ \iota_{1} :  \iota_{1}^{-1}(\iota_{1}(M)\cap U)\rightarrow \iota_{1}^{-1}(\iota_{1}(M)\cap U)
    \end{equation*}
    is the identity map in $\iota_{1}^{-1}(\iota_{1}(M)\;\cap \;U_{1})\subset M$. In \cite{vdbs} the second author showed  (Lemma 64 in \cite{vdbs}) that this induces a partial ordering on $\mathcal{I}_{\textrm{ext}}$.

\end{itemize}
These proofs are relatively straightforward but tedious to write down precisely (as one has to constantly change the neighborhoods one is working on). This changing of neighborhoods becomes an issue when considering an (uncountable) totally ordered subfamily of equivalence classes of extensions $\{[\ia]\}_{\alpha\in\mathcal{A}}$ and trying to construct an upper bound for this subfamily by 'gluing' together all $U_{\alpha}\cup\ia(M)$ and identifying points appropriately, as we already discussed in the paragraph above Section \ref{sec:topologicalremarks}. 
Let us remark that the relations ''$=_{\pr}$" and ''$\leq_{\pr}$" are compatible with Definitions \ref{def:equiv of Ns} and \ref{def:order on Ns} above: If $(\Mextone,\iota_{1}),(\Mexttwo,\iota_{2})$ have empty past boundary and $M$ is globally hyperbolic, then $(\Mextone,\iota_{1})=_{\pr}(\Mexttwo,\iota_{2})$, resp.\ $\leq_{\pr}$, then the corresponding future boundary extensions $N_1$ and $N_2$ satisfy $N_1\cong N_2$, resp.\ $N_1\lesssim N_2$ (note that even if the future boundary extensions are not regular, the relations $\cong $ and $\lesssim$ are well defined).

\end{remark}

\subsection{Existence of set theoretic maximal elements}\label{subsec:settheoreticmax} 

In this section we will show that the set of equivalence classes of regular future g-boundary extensions $\Iset$ of a given spacetime $(M,g)$ contains at least one set-theoretic maximal element proceeding via a standard Zorn's lemma proof. In other words, we show that there exist upper bounds for any arbitrary totally ordered subset $\Jset=\{[N_{\alpha}]\}_{\alpha\in \mathcal{A}}\subset \Iset$ of regular future g-boundary extensions.

This is organized as follows. First, a candidate for a representative of an upper bound, $\Nupp$, for any such totally ordered $\Jset$ is constructed by gluing together all the $\Na$'s: We take a disjoint union, identify points via the embeddings $\psiab$ from the ordering relation and take $\Nupp$ to be the quotient space (with the quotient topology) and define a natural map $\iupp :M\to \Nupp$ via $\iupp=\pi\circ \ia$ for any $\ia$. 

Next, we need to show that this quotient space belongs to $\Iset$, i.e. is itself a regular future g-boundary extension for $M$. As quotient topologies are in general quite badly behaved, especially with respect to separation axioms and potentially second countability (if $\{[N_{\alpha}]\}_{\alpha\in \mathcal{A}}$ is not countable), some care is necessary. The order relation straightforwardly gives us that $\pi$ is an open map (cf.\ Lemma \ref{lem:pi_Nmax_open}) which implies that $\Nupp $ is indeed Hausdorff (cf. Lemma \ref{lem:Nmax_Hausdorff}) and that $((\Nupp, \tau_{q}),\iupp)$ is a candidate for a future boundary extension. For second countability we first establish the ''regularity'' part of Definition \ref{def:regular-g-boundary-ext}, i.e., we show that $((\Nupp, \tau_{q}),\iupp)$ is  a candidate for a regular future g-boundary extension (cf. Lemma \ref{lem:Nmax_regular_candidate}). Once this is done, second countability follows from Lemma \ref{lem:AA2}. Lastly, openness of $\pi$ straightforwardly allows us to project charts for $\Na$ onto the quotient $\Nupp$ showing that it is indeed a topological manifold with boundary (cf. Proposition \ref{prop:Nmax_is_reg_future_g_boundary_ext}). 

Finally, that $[\Nupp]$ indeed is an upper bound for the totally ordered subset of extensions $\Jset=\{[N_{\alpha}]\}_{\alpha\in \mathcal{A}}$ follows directly from our construction and hence Zorn's Lemma implies that the partially ordered subset $\mathcal{I}$ has a maximal element.

So, let $(M,g)$ be a $C^{2}$ spacetime and let $\mathcal{J} \coloneqq\{[\Na]\}_{\alpha\in\mathcal{A}}$ for some index set $\mathcal{A}$ be a totally ordered subset of equivalence classes of regular future g-boundary extensions. We choose 
any family $\{\Na\}_{\alpha\in\mathcal{A}}$ of representatives and set

\begin{equation}
    \label{eq:M'}
    N'\coloneqq \bigsqcup_{\alpha\in \mathcal{A}}\Na
\end{equation}
\begin{equation}
    \label{eq:Mmax}
    \Nupp\coloneqq N'/\sim,
\end{equation}
where the equivalence relation $\sim$ is defined as follows for two arbitrary points $p\in \Na \subset N'$ and $q\in \Nb \subset N'$: 
\begin{equation}
\label{eq:equivalencerelation}
\begin{split}
    p\sim q 
\end{split}
\quad:\;\longleftrightarrow\quad
\begin{split}
    \begin{cases}
        q=\psiab(p) \quad\textrm{if} \;\Na\lesssim\Nb\\
        p=\psiba(q) \quad\textrm{if} \;\Nb \lesssim \Na
    \end{cases}
\end{split}
\end{equation}
Note that this is indeed an equivalence relation and implies $p\sim q \iff p=q$ if $p,q\in \Na$. We will denote the quotient map by $\pi$, i.e., $\pi : N'\rightarrow \Nupp$, $p\mapsto \pi(p)$.  We endow $\Nupp$ with the quotient topology $\tau_{q}$, i.e., $U\subset\Nupp$ is open if and only if $\pi^{-1}(U)\subset N'$ is open. Next we define a map $\iupp : M\to \Nupp$ via
\begin{align}
	\iupp(p) &\coloneqq \pi(\ia(p)) 
\end{align}
where $\alpha \in \mathcal{A}$ can be chosen arbitrarily as $\pi(\ia(p)) =\pi(\ib(p)) $ for any $\alpha,\beta \in \mathcal{A}$ since $\psiab=\ib\circ \ia^{-1}$ on $\ia(M)$ by definition of $\lesssim$.
Note that a priori both $N'$ and $\Nupp$ may depend on the choices of representatives, but this doesn't bother us for the moment as we only aim to show the existence but not necessarily uniqueness of a  regular future g-boundary extension $\Nupp$ for which $[\Nupp]$ is an upper bound for the totally ordered set $\Jset$. However, we will see in Remark \ref{rem:equiv_class_ofNmax_unique} that the equivalence class $[(\Nupp,\iupp)]$ obtained from this process is in fact independent of the chosen representatives.\\

\begin{lemma}\label{lem:pi_Nmax_open}
$\pi: N'\to (\Nupp,\tau_q)$ is an open map. 
\end{lemma}
\begin{proof}  Fix $\alpha\in \mathcal{A}$ and any open $U_{\alpha}\subset\Na$. Then for any $\beta\in \mathcal{A}$ we have either  $\Nb \lesssim \Na$, $\Nb= \Na$ or $\Na \lesssim \Nb$. In all these cases we will show that $(\pi|_{\Nb})^{-1}(\pi(U_{\alpha}))$ is open.  This will imply that $\pi^{-1}(\pi(U_{\alpha}))=\bigsqcup_{\beta\in \mathcal{A}} (\pi|_{\Nb})^{-1}(\pi(U_{\alpha}))$ is open in $N'$, i.e., $\pi(U_{\alpha})$ is open in $(\Nupp, \tau_q)$ implying that $\pi$ is an open map (as both $\alpha$ and $U_{\alpha}$ were arbitrary).

Assume $\Nb\lesssim \Na$ and let $\psiba : \Nb\rightarrow\Na$ be the embedding. Then, by definition of $\pi$, $\pi|_{\Nb}(q)=\pi|_{\Na}(p)$ for $p\in \Na,q\in \Nb$ if and only if $p=\psiba(q)$, so $(\pi|_{\Nb})^{-1}(\pi(U_{\alpha}))=\psiba^{-1}(U_{\alpha})$ which is open in $\Nb$. In case $\Nb=\Na$, clearly $(\pi|_{\Nb})^{-1}(\pi(U_{\alpha}))=U_{\alpha}$ is open. In case  $\Na \lesssim\Nb$, $\pi|_{\Nb}(q)=\pi|_{\Na}(p)$ for $p\in \Na,q\in \Nb$ if and only if $q=\psiab(p)$, so $(\pi|_{\Nb})^{-1}(\pi(U_{\alpha}))=\psiab(U_{\alpha})$ which is open in $\Nb$.
 So in summary $\pi$ is an open map.
 \end{proof}

Continuity and openness of $\pi$ together with injectivity of $\pi\big|_{\Na}$ now immediately imply that $\iupp=\pi\circ \ia=\pi\big|_{\Na}\circ \ia $ (for any $\alpha\in \mathcal{A}$) is a topological embedding onto the open set $\iupp(M)=\pi(\ia(M))$. 
  Further, $$\Nupp=\pi(N')=\pi\big(\bigsqcup_{\alpha\in\mathcal{A}}\overline{\ia(M)}^{\Na}\big)=\overline{\pi\big(\bigsqcup_{\alpha\in\mathcal{A}} \ia(M)\big)}^{\tau_q}=\overline{\iupp(M)}^{\tau_q}$$ by continuity of $\pi$. So $((\Nupp, \tau_q),\iupp)$ is a candidate for a future boundary extension of $(M,g)$ as in Definition \ref{def:candidate_boundary_ext} and we may define a family of timelike thickenings for $\Nupp$ as in Definition \ref{def:thickenings}. Our next crucial step is to show that $((\Nupp, \tau_q),\iupp)$ is a candidate for a regular future g-boundary extension of $(M,g)$ as in Definition \ref{def:candidate_regular_boundary_ext}. We start with the following Lemma

\begin{lemma} \label{lem:thickenings_under_pi} Let $X\in T_tM, r>0$. We have $$\pi^{-1}(O^{\Nupp}_{X,r})=\bigsqcup_{\alpha\in \mathcal{A}} O^{\Na}_{X,r}, $$ where $O^{\Na}_{X,r}$ denotes the timelike thickening corresponding to $X,r$ in $\Na$. Further, for any $\tau_q$-open $V\subset \Nupp$ and any future directed timelike geodesic $\gamma:[0,1)\to M$  with $p=\lim^{(\Nupp,\tau_q)}_{t\to 1^-} \iupp(\gamma(t))\in V\setminus \iupp(M)$  there exist $\alpha\in\mathcal{A}$ and $n,m\in\mathbb{N}$ such that $\pi(O^{\Na}_{\dot{\gamma}(1-\frac{1}{n}),\frac{1}{m}})= O^{\Nupp}_{\dot{\gamma}(1-\frac{1}{n}),\frac{1}{m}}$ and $O^{\Nupp}_{\dot{\gamma}(1-\frac{1}{n}),\frac{1}{m}}\subset V$.
\end{lemma}

 \noindent \textit{Proof.}  We first show \begin{equation}\label{eq:pi(Na)subsetpi(Nmax)}
     	\pi(O^{\Na}_{X,r})\subset O^{\Nupp}_{X,r}
     \end{equation} for all $\alpha\in\mathcal{A}$: Let $p\in O^{\Na}_{X,r}$. Either $p=\ia(\gamma_Y(t))\in \mathrm{int}(\Na) $ for some $Y\in B_r(X)\subset TM$ and $t\in (0,b_Y)$. 
    Then $\pi(p)=\iupp(\gamma_Y(t))$ by definition of $\iupp$ and hence $\pi(p)\in O^{\Nupp}_{X,r}$. Or $p=\lim^{\Na}_{t\to b_Y^-} \ia(\gamma_Y(t))$. Then, by continuity of $\pi$, $\pi(p)=\lim^{\Nupp}_{t\to b_Y^-} \pi(\ia(\gamma_Y(t)))=\lim^{\Nupp}_{t\to b_Y^-} \iupp(\gamma_Y(t)) \in O^{\Nupp}_{X,r}$. 
    This implies $$\pi^{-1}(O^{\Nupp}_{X,r})\supset \bigsqcup_{\alpha\in \mathcal{A}} O^{\Na}_{X,r}.$$
    
    To show the other inclusion let $p\in O^{\Nupp}_{X,r}$. If $p= \iupp(\gamma_Y(t)) $ for some $t\in(0,b_Y)$, then clearly $\pi^{-1}(p)=\bigsqcup_{\alpha\in \mathcal{A}} \{ \ia(\gamma_Y(t)) \} \subset \bigsqcup_{\alpha\in \mathcal{A}} O^{\Na}_{X,r}$. So assume $p=\lim^{(\Nupp,\tau_q)}_{t\to b_Y^-} \iupp(\gamma_Y(t)) $. Let $q\in \pi^{-1}(p)$. Then we must have $q\in \Na$ for \emph{some} $\alpha\in \mathcal{A}$. For any open neighborhood $U$ of $q$ in $\Na$ we have that $\pi(U)=\pi\big|_{\Na}(U)$ is a $\tau_q$-open (because $\pi$ is an open map, see Lemma \ref{lem:pi_Nmax_open}) neighborhood of $p$ in $\Nupp$, so $\iupp(\gamma_Y(t))$ must be contained in $\pi\big|_{\Na}(U)$ for all $t$ sufficiently close to $b_Y$. This means that $\ia(\gamma_Y(t))=(\pi\big|_{\Na})^{-1}(\iupp(\gamma_Y(t)))\in U$ for all $t$ sufficiently close to $b_Y$, 
    so $\lim^{\Na}_{t\to b_Y^-} \ia(\gamma_Y(t))$ exists and equals $q$. Thus $q\in O^{\Na}_{X,r}$, i.e., 
    $$ \pi^{-1}(O^{\Nupp}_{X,r})\subset \bigsqcup_{\alpha\in \mathcal{A}} O^{\Na}_{X,r}. $$

\begin{wrapfigure}{r}{0.5\textwidth}
\centering \includegraphics[width=0.48\textwidth]{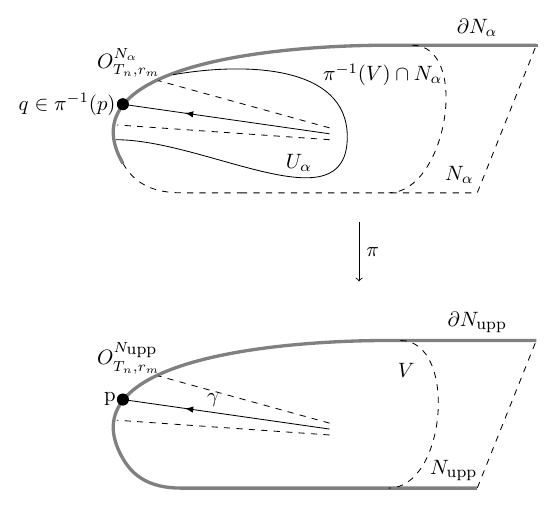}
  \caption{While the boundary of $\Nupp$ may be ''larger'' than the boundary of any of the $\Na$'s, we can lift any point $p\in \Nupp \setminus \iupp(M)$ to some $\Na$ and since $\Na$ is a manifold with boundary there exists a suitable relatively compact neighbourhood $U_{\alpha}$ such that any $O^{\Na}_{X,r}\subset U_{\alpha}$ will satisfy $\pi(O^{\Na}_{X,r})=O^{\Nupp}_{X,r}$.}
  \label{fig:1}
\end{wrapfigure}

  Let us now prove the second part. Let $V\subset \Nupp$ be $\tau_q$-open, fix $p\in V\setminus \iupp(M)$ and let $\gamma:[0,1)\to M$ be a future directed timelike geodesic with $p=\lim_{t\to 1^-} \iupp(\gamma(t))$. Choose any $\alpha \in \mathcal{A}$ for which $\pi^{-1}(p)\cap \Na \neq \emptyset$ and let $q\in \Na$ such that $\pi(q)=p$. Since $\pi^{-1}(V)$ is open in the disjoint union $N'$, $\pi^{-1}(V)\cap \Na$ 
  is open in $\Na$. Because $\Na$ is a topological manifold (with boundary) we can find an open $U_\alpha\subset \Na$ such that $q\in U_\alpha$ and the closure of $U_\alpha$ in $\Na$ is compact and contained in $\pi^{-1}(V)\cap \Na $. Set $T_n:=\dot{\gamma}(1-\frac{1}{n}) $, $r_m=\frac{1}{m}$ and let $n,m\in \mathbb{N}$ be such that $ O_{T_n,r_m}^{\Na}\subset U_\alpha$ (such an $n,m$ exists because $\Na$ is a regular future g-boundary extension), cf.~Figure \ref{fig:1}.

    We want to show that $\pi(O^{\Na}_{T_n,r_m})= O^{\Nupp}_{T_n,r_m}$: We already know $\pi(O^{\Na}_{T_n,r_m})\subset O^{\Nupp}_{T_n,r_m}$ from \eqref{eq:pi(Na)subsetpi(Nmax)}. So assume to the contrary that there exists $x\in O^{\Nupp}_{T_n,r_m}\setminus \pi(O^{\Na}_{T_n,r_m})$. Since $x\in O^{\Nupp}_{T_n,r_m}$ there exists $\gamma_Y$ with $Y\in B_{r_m}(T_n)$ and $0<t_0\leq b_Y$ such that $\lim^{(\Nupp,\tau_q)}_{t\to t_0^-} \iupp\circ\gamma_Y(t) =x$. 
    Since $(\ia\circ\gamma_Y)|_{(0,b_Y)}\subset O_{X,r}^{\Na}\subset U_\alpha$, relative compactness of $U_\alpha$ guarantees that there exists a sequence $t_k\in (0,b_Y)$ with $ t_k \to t_0$ for which $(\ia\circ\gamma_Y)(t_k)\to y$ in $\Na$ for some $y\in \overline{U_\alpha}^{\Na}\subset \pi^{-1}(V)\cap \Na$. By continuity of $\pi$ and definition of $\iupp$, we would have $\pi(y)=x$. Since we assumed $x\in  O^{\Nupp}_{T_n,r_m}\setminus \pi(O^{\Na}_{T_n,r_m})$, we would have $y\in \Na\setminus O^{\Na}_{T_n,r_m}$. However, by definition of $O^{\Na}_{T_n,r_m}$, whenever $\lim^{\Na}_{t\to t_0^-} (\ia\circ \gamma_Y)(t)$ exists in $\Na$ for any $0<t_0\leq  b_Y$, then this limit will belong to $O^{\Na}_{T_n,r_m}$. So $\lim^{\Na}_{t\to t_x^-} (\ia\circ \gamma_Y)(t)$ cannot exist. 
    Hence there must exist a different sequence $t_k'$ for which $(\ia\circ\gamma_Y)(t_k')\to y'\neq y$ in $\Na$ (the diverging case can again be excluded by relative compactness). By continuity of $\pi$ we must again have $x=\pi(y')$ but $\pi\big|_{\Na}$ is injective, so $y=y'$, giving a contradiction.

    Hence we indeed have $O^{\Nupp}_{T_n,r_m}=\pi(O^{\Na}_{T_n,r_m})$. But then clearly $O^{\Nupp}_{T_n,r_m}=\pi(O^{\Na}_{T_n,r_m}) \subset \pi(U_\alpha )\subset V$ and we are done.\qed\\

\begin{remark}	\label{remarkquotient}
Note that the proof only used injectivity of $\pi\big|_{\Na}: \Na \to \Nupp$, compatibility of the embedding $\iupp$ with $\pi$ and $\ia$ for all $\alpha \in \mathcal{A}$, i.e., that $\iupp=\pi\circ\ia$ for any $\alpha \in \mathcal{A}$, and that $\pi$ is an open map. Importantly we did neither use that the family $\{[\Na]\}_{\alpha\in\mathcal{A}}$ was totally ordered nor any further details on the definition of the equivalence relation. Hence we will be allowed to use this fact (and any results deriving directly from it) in the construction in next section, Section \ref{subsec:actualmax}, as well. 
\end{remark}

 \begin{lemma}\label{lem:Nmax_regular_candidate}
     $((\Nupp, \tau_q),\iupp)$ is a candidate for a regular future g-boundary extension of $(M,g)$.
 \end{lemma}
 \begin{proof}

    That there exists a future directed timelike geodesic $\gamma:[0,1)\to M$ with $p=\lim_{t\to 1^-} \iupp(\gamma(t))$  for any $p\in \Nupp \setminus \iupp(M)$ follows immediately from the construction: For $p\in \Nupp \setminus \iupp(M)$, there exists $\alpha\in \mathcal{A}$ and $\pa \in \Na\setminus \ia(M)$ such that $\pi(\pa)=p$. So, since $\Na$ is a regular future g-boundary extension there exists  a future directed timelike geodesic $\gamma:[0,1)\to M$ with $\pa=\lim_{t\to 1^-} \ia(\gamma(t))$. So $p=\lim_{t\to 1^-} \iupp(\gamma(t))$ follows from the definition of $\iupp$ and continuity of $\pi$.
           
  It remains to show that all timelike thickenings $O^{\Nupp}_{X,r}$ are open and that for any future directed timelike geodesic $\gamma:[0,1)\to M$ with $p=\lim_{t\to 1^-} \iupp(\gamma(t))\in \Nupp\setminus \iupp(M)$ the collection 
  $\{O^{\Nupp}_{\dot{\gamma}(1-\frac{1}{n}),\frac{1}{m}}: n,m\in \mathbb{N}\}$ is a neighborhood basis for $p$. This follows immediately from the previous Lemma \ref{lem:thickenings_under_pi}.
  \end{proof}

 \begin{remark}\label{rem:regular_ok_whenever_quotient_compatible_and_pi_open}
    Again, the proof only uses injectivity of the $\pi\big|_{\Na}$, compatibility of the embedding $\iupp$ with $\pi$ and $\ia$ for all $\alpha \in \mathcal{A}$ and that $\pi$ is an open map. Hence we will be allowed to use this fact in the construction in next section, Section \ref{subsec:actualmax}, as well and will in fact do so to obtain Lemma \ref{lem:Ntilde_regular_candidate}.
 \end{remark}

Let us now turn towards topology. As already pointed out Lemma \ref{lem:AA2} immediately gives second countability. We next show Hausdorffness.

\begin{lemma}\label{lem:Nmax_Hausdorff}
The topological space $(\Nupp,\tau_q)$ is Hausdorff.
\end{lemma}
\begin{proof}	
 Take any $q_{1},q_{2}\in\Nupp$ such that $q_1\neq q_2$. Then there exist two points $p_{\alpha}\in\Na$ and $p_{\beta}\in \Nb$ such that $q_{1}=\pi(p_{\alpha})$ and $q_2=\pi(p_{\beta})$. Assume w.l.o.g.\ that $\Na \lesssim \Nb$. Then since $q_1\neq q_2$, also $p_{\beta}\neq \psiab(p_{\alpha})$. As $\Nb$ is Hausdorff, there exist disjoint open neighborhoods $U_{1}$ and $U_{2}$ of $p_{\beta}$ and $\psiab(p_{\alpha})$. Define the subsets $V_{1}\coloneqq\pi(U_{1})=\pi|_{\Nb}(U_{1})$ and $V_{2}\coloneqq \pi(U_{2})=\pi|_{\Nb}(U_{2})$ which satisfy that $q_{1}\in V_{1}$ and $q_{2}\in V_{2}$. By Lemma \ref{lem:pi_Nmax_open} both $V_1$ and $V_2$ are open and invertibility of $\pi|_{\Nb}$ together with disjointedness of $U_{1}$ and $U_{2}$ implies that $V_{1}\cap V_{2}=\emptyset$. 
\end{proof}
 
We are now ready to equip $\Nupp$ with suitable charts turning it into a topological manifold with boundary and put everything together.

\begin{proposition}\label{prop:Nmax_is_reg_future_g_boundary_ext}
    Let $(M,g)$ be a $C^{2}$ spacetime and $\mathcal{J}$ a totally ordered set of of regular future g-boundary  extensions. Then $\Nupp$ is a regular future g-boundary extension of $(M,g)$.
\end{proposition}

\begin{proof}
 
Thanks to Lemmas \ref{lem:Nmax_regular_candidate} and \ref{lem:Nmax_Hausdorff}, it only remains to show that $(\Nupp,\tau_q)$ carries the structure of a topological manifold with boundary, i.e., that there exist suitable charts. 
The idea is to construct charts on $\Nupp$ using the charts on $\Na$ (for each $\alpha\in\mathcal{A}$) and composing them with the quotient map $\pi$. 
Take a point $p\in \Nupp$, take $\alpha\in \mathcal{A}$ and $p_{\alpha}\in \Na$ such that $p=\pi(p_{\alpha})$ and a coordinate chart $(U_{\alpha},x_{\alpha})$ around $p_{\alpha}$ in $\Na$. Note that if $p_{\alpha}\in\ia(M)=\mathrm{int}(\Na)$, $x_{\alpha}$ is a homeomorphism onto an open subset of $\mathbb{R}^d$, while if $p_{\alpha}\in \Na\setminus \ia(M)$, $x_{\alpha}$ is a homeomorphism onto an open set in the half space $[0,\infty)\times \mathbb{R}^{d-1}$. As $\pi$ is an open map, $\pi(U_{\alpha})$ is an open neighborhood of $p$ in $\Nupp$. Then, on $\Nupp$ we define the map $x : \pi(U_{\alpha})\rightarrow \mathbb{R}^{d}, p\mapsto x_{\alpha}(\pi\big|_{\Na}^{-1}(p))$, noting that $\pi\big|_{\Na}: U_{\alpha}\to \pi(U_{\alpha})$ is a bijection. In the following it will be proven that $(\pi(U_{\alpha}), x)$ is a coordinate chart for $\Nupp$.
    \par We show that the map $x$ is bicontinuous. By definition of $x$ it holds that $x^{-1}(W)=\pi(x_{\alpha}^{-1}(W))$ for all open $W\subset x(\pi(U_{\alpha}))\subset\mathbb{R}^{d}$, which is an open set as $x_{\alpha}^{-1}(W)$ is open (since $x_{\alpha}$ is continuous) and $\pi$ is an open map. Hence, $x$ is continuous. 
    In order to see that $x^{-1}$ is continuous, simply note that $x^{-1}=\pi \circ x_{\alpha}^{-1}$ is the composition of continuous maps. Since we only need a topological manifold, there is no further compatibility between charts we'd have to check.
\end{proof}

\begin{remark}\label{rem:equiv_class_ofNmax_unique} Let us at this point remark that while $(\Nupp,\iupp)$ might depend on the chosen representatives $(\Na,\ia)$ of $[(\Na,\ia)]$, its equivalence class $[(\Nupp,\iupp)]$, which is now well-defined as we just established that $\Nupp$ is a regular future g-boundary extension, does not: For every $\alpha$ let $(\Na,\ia)$ and $(\Na',\ia')$ be two regular future g-boundary extensions with $[(\Na,\ia)]=[(\Na',\ia')]$ and consider 
$$\psi: \bigsqcup_{\alpha\in \mathcal{A}} \Na \to    \Nupp' :=(\bigsqcup_{\alpha\in \mathcal{A}} \Na')/\sim  $$
defined by $\psi(\pa):=\pi'(\psi_{\alpha\alpha'}(\pa))$, where $\psi_{\alpha\alpha'}$ is the homeomorphism arising from the equivalence relation $(\Na,\ia)\cong (\Na',\ia')$ (with $\pi'$ the projection $\pa \mapsto [\pa]'\in \Nupp'$ for $\Nupp'$) for $\pa\in \Na\subset \bigsqcup_{\alpha} \Na $. Clearly this is well defined, surjective and satisfies $\psi(\pa)=\psi(\pb)$  for $\pa,\pb $ with $[\pa]=[\pb]$ (noting that for $(\Na,\ia)\lesssim (\Nb,\ib)$ also $(\Na',\ia')\lesssim (\Nb',\ib')$ and $\psi_{\beta \beta'}|_{\psiab(\Na)}=\psi_{\alpha'\beta'}\circ \psi_{\alpha \alpha'}\circ \psiab^{-1}$ since all $\psi_{ij}$ are uniquely determined from $\iota_j\circ\iota_i^{-1}$ by extending continuously). So by the universal property of the quotient space there exists a well-defined continuous and surjective map
\begin{align*}
    \psi_{\Nupp\Nupp'}: \Nupp &\to    \Nupp'  \\
    \pi(\pa) &\mapsto \psi(\pa)=\pi'(\psi_{\alpha\alpha'}(\pa)).
\end{align*}
Analogously, just switching the roles of $\Nupp$ and $\Nupp'$, we obtain a 
continuous and surjective map
\begin{align*}
  \psi_{\Nupp' \Nupp}: \Nupp' &\to   \Nupp \\
  \pi'(\pa) &\mapsto \tilde{\psi}(\pa)=\pi(\psi_{\alpha'\alpha}(\pa)).  
\end{align*}
By construction (using that $\psi_{\alpha \alpha'}^{-1}=\psi_{\alpha' \alpha}$) we have $\psi_{\Nupp' \Nupp}\circ \psi_{\Nupp \Nupp'}=\mathrm{id}_{\Nupp}$ and $\psi_{\Nupp \Nupp'}\circ \psi_{\Nupp' \Nupp}=\mathrm{id}_{\Nupp'}$, so $\psi_{\Nupp \Nupp'}$ and $\psi_{\Nupp' \Nupp}$ are homeomorphisms and, again by construction, $\psi_{\Nupp \Nupp'}\circ \iupp =\iupp'$. Hence, $(\Nupp, \iupp)\cong (\Nupp', \iupp')$.
\end{remark}

Since, by construction, $[\Nupp]$ is an upper bound for $\Jset=\{[\Na]\}_{\alpha\in \mathcal{A}}$, we have thus established

\begin{theorem}\label{theo:maxextension_setheoretic}Let $(M,g)$ be a $C^{2}$ spacetime and $\mathcal{J}$ a totally ordered set of of regular future g-boundary  extensions. Then there exists an upper bound for $\mathcal{J}$, i.e.\ there exists $[N]\in \Iset$ such that $[\Na]\leq [N]$ for any $[\Na]\in \Jset$.\end{theorem}

Let us observe that this immediately gives the following Corollary.

\begin{corollary}\label{cor:existence_of_set_theoretic_max} Let $(M,g)$ be a $C^{2}$ spacetime and $\mathcal{I}$ a partially ordered set of of regular future g-boundary  extensions. Then there exists a maximal element for $\mathcal{I}$, i.e.\ there exists $[N]_{\mathrm{max}}\in \Iset$ which satisfies that if $ [N]_{\mathrm{max}}\leq [N] $ for any $[N]\in \Iset$ one must already have equality $[N]_{\mathrm{max}}=[N]$\end{corollary}
\begin{proof}
	This follows directly from the existence of upper bounds for every totally ordered subset $\Jset$ and Zorn's Lemma.
\end{proof}

Of course, such set theoretic maximal elements are expected to be non-unique. For instance, we believe the two inequivalent extensions of the Misner and Taub-NUT spacetimes (described in \cite{Hawking} and \cite{Chrusciel}, cf.\ in particular Prop.\ 5.16, respectively) to both be maximal elements. However a proof of this in our case is less immediate than the corresponding proof of \cite[Prop.\ 5.16]{Chrusciel} due to the non-constructive nature of Zorn's lemma.

\section{Existence of a unique maximal regular future g-boundary extension}\label{subsec:actualmax}

The example of Misner spacetime given in \cite{Misner} suggests that uniqueness necessitates an additional condition to be imposed on $(M,g)$. Together with the work by Chru\'sciel on uniqueness of conformal boundaries, \cite{Chrusciel}, it seems natural to consider the following additional condition

\begin{definition}[Intertwined timelike geodesics]
\label{defintertwined} Let $\gamma_{1}: [0,b_{Y_1})\rightarrow M$, $Y_1:=\dot{\gamma}_1(0)$, and $\gamma_{2} : [0,b_{Y_2})\rightarrow M$, $Y_2:=\dot{\gamma}_2(0)$, be two future directed future inextendible timelike geodesics in a $C^2$ spacetime $(M,g)$. 
Then, we say that $\ga$ and $\gb$ \emph{are not intertwined} provided one of the following conditions holds:

\begin{enumerate}[label=(\roman*)]
    \item For any radii $r>0$, $\rho> 0$  there exist $s_{1}\in (0,b_{Y_1})$, $s_{2}\in\left(0,b_{Y_2}\right)$ such that $\ga(\left[s_{1},b_{Y_1})\right)\subset O^M_{Y_2, \rho}$
    and $\gb(\left[s_{2},b_{Y_2})\right)\subset O^M_{Y_1,r}$. 
    \item There exists $s_{1}\in (0,b_{Y_1})$, $s_{2}\in\left(0,b_{Y_2}\right)$ and radii $r,\r >0$ such that $O^M_{\dot{\gamma}_{1}(s_1),r} \cap O^M_{\dot{\gamma}_{2}(s_2),\r}=\emptyset$. 
\end{enumerate}
If neither of these conditions hold, then we say that $\ga$ and $\gb$ \emph{are intertwined}. 
\end{definition}

\par Heuristically, two curves $\ga$ and $\gb$ are not intertwined if they \emph{merge}, i.e. they approach each other and remain arbitrarily close (case (i) of the previous definition), or \emph{part}, i.e. there exists a fixed distance at which these curves will, as long as defined, never be (case (ii) of the previous definition). In other words, it is not possible that these geodesics come arbitrarily close to each other without remaining close afterwards. Intertwined geodesics, as pointed out by Chru\'sciel \cite{Chrusciel} and Sbierski \cite{Sbierski2022}, appear for example in the Taub-NUT or Misner spacetime and lead to the existence of distinct
extensions of the original spacetime. 
In particular, a ''common'' extension of two arbitrary (i.e., non-ordered) regular future g-boundary extensions $\Na$ and $\Nb$ might fail to be Hausdorff if there exist intertwined timelike geodesics in $M$.

Before proceeding let us remark that condition (i) in Definition \ref{defintertwined} could be rewritten using $O^M_{\dot{\gamma}_2(b_{Y_2}-\frac{1}{n_2}),\rho}$ and $O^M_{\dot{\gamma}_1(b_{Y_1}-\frac{1}{n_1}),r}$ instead of $O^M_{Y_2,\rho}$ and  $O^M_{Y_1,r}$ for any $n_2, n_1$.
\begin{lemma}\label{lem:propertiesthickenings}
  Let $(M,g)$ be a strongly causal spacetime with $C^2$-metric $g$ and $\gamma:[0,b)\to M$ an inextendible future directed timelike geodesic in $M$. Then, if $\gamma':[0,b')\to M$ is any inextendible future directed timelike geodesic in $M$ such that the pair $\gamma,\gamma'$ satisfies point (i) in Definition \ref{defintertwined}, then for any $n\in \mathbb{N}, r>0$ there exists $s'\in (0,b')$ such that $\gamma'([s',b'))\subset O^M_{\dot{\gamma}(b-\frac{1}{n}),r}$. 
\end{lemma} 
\begin{proof} 

Fix $r>0,n\in \mathbb{N}$ and set $T_n:=\dot{\gamma}(b-\frac{1}{n})$. Choose $\bar{\r}(r,n)>0$ such that $\{\dot{\gamma}_Y(b-\frac{1}{n}): Y\in B_{\bar{\r}(r,n)}(\dot{\gamma}(0))\}\subset B_r(T_n)$, noting that such a $\bar{\r}(r,n)$ exists by continuous dependence of tangents to geodesics on the initial data. Then 
$$ O^M_{\dot{\gamma}(0),\bar{\r}(r,n)} \setminus O^M_{T_n,r} 
\subset \{\gamma_Y([0,b-\frac{1}{n}]): Y\in \overline{B_{\bar{\r}(r,n)}(\dot{\gamma}(0))\cap T_tM}\}.$$

Note that the latter set is compact. Now if $ \gamma'([s',b'))\not\subset O^M_{T_n,r}$ for any $s'\in(0,b')$, but by point (i) in Definition \ref{defintertwined} there exists $\bar{s}'\in(0,b')$ such that $ \gamma'([\bar{s}',b'))\subset O^M_{\dot{\gamma}(0),\bar{\r}(r,n)}  $, then there is a sequence $s_k\to b$ for which 
$$\gamma'(s_k)\in O^M_{\dot{\gamma}(0),\bar{\r}(r,n)} \setminus  O^M_{T_n,r} $$
So by the above $\gamma'(s_k)$ is contained in a compact set for all $k$. This shows that $\gamma'$ is an inextendible timelike curve partially imprisoned in a compact set, contradicting  strong causality of $(M,g)$ (see e.g. \cite{Minguzzi2008} Prop. 2.5)
\end{proof}

In the remainder of this section we show that, indeed, $(M,g)$ not containing any intertwined future directed timelike geodesics is a sufficient condition for the existence of a maximal  regular future g-boundary extension (provided that $(M,g)$ is regular future g-boundary extendible at all).

\begin{definition}\label{def:Maximal_regular_extension}
    A regular future g-boundary extension $(\Nmax,\ima)$ of $(M,g)$ is said to be a maximal regular future g-boundary extension if any other regular future g-boundary extension $(N,\iota)$ satisfies $[N]\leq [\Nmax]$. 
\end{definition} 

\begin{remark}\label{rem:maxextension}
\begin{enumerate}
    \item   By this definition any maximal regular future g-boundary extension automatically has to be unique in the following sense: If $(\Nmax,\ima)$ and $(\hat{N}_{\mathrm{max}},\hat{\iota}_{\mathrm{max}})$ are two maximal regular future g-boundary extensions, then $[\Nmax]= [\hat{N}_{\mathrm{max}}]$, i.e., there exists a homeomorphism between them which, when pulled back by the embeddings $\ima$ resp.\ $\hat{\iota}_{\mathrm{max}}$, gives the identity on $M$. 
    \item   Clearly the equivalence class of any maximal regular future g-boundary extension has to be a maximal element for the partially ordered set $$\Iset=\{[N]: (N,\iota)\; \mathrm{is\,a\,regular\,future\,g-boundary\,extension}\}.$$ However, a  set theoretic maximal element of $\Iset $ need not satisfy that its representatives are maximal regular future g-boundary extensions in the sense of the above Definition \ref{def:Maximal_regular_extension}. 
    \item If any two set theoretic maximal elements $[N]_{\mathrm{max}}$ and $[N]'_{\mathrm{max}}$ for $\Iset$  are equal, then any representative for their equivalence class is a maximal regular future g-boundary extension.
\end{enumerate}
\end{remark}

\begin{definition}
    Let $(M,g)$ be a $C^{2}$ spacetime. It is called \emph{regular future g-boundary extendible} if the set the set of regular future g-boundary extensions of $(M,g)$ is non-empty.
\end{definition}

For instance, a sufficient condition for a $C^{2}$ globally hyperbolic $(M,g)$ to be regular future g-boundary extendible is that there exists a $C^2$ spacetime extension (in the usual sense, cf. Definition \ref{def:classicalextension}) with empty past boundary (cf. Section \ref{sec:compatibility}).

As mentioned, our goal of this section is to show that there exists a maximal  regular future g-boundary extension if $(M,g)$ does not contain any intertwined future directed timelike geodesics. The strategy of the proof proceeds as follows: We first show that if $(M,g)$ does not contain any intertwined timelike geodesics we can, essentially, do the same construction as in the previous section for \emph{any} two regular future g-boundary extensions $\Na$ and $\Nb$. That is, if we define
$N:=\Na \sqcup \Nb / \sim $ for an appropriate equivalence relation, then $N$ naturally becomes a regular future g-boundary extension. Our strategy essentially follows the one in Section \ref{subsec:settheoreticmax}, and we are even are able to make direct use of some of the results from that section, such as Lemmas \ref{lem:thickenings_under_pi} and \ref{lem:Nmax_regular_candidate} (cf.\ Remarks \ref{remarkquotient} and \ref{rem:regular_ok_whenever_quotient_compatible_and_pi_open}). However showing openness of the quotient map $\pi$ and Hausdorffness of the quotient topology (cf. Lemma \ref{lem:NtildeHausdorff}) becomes much more involved (and for both our proofs rely on not having intertwined timelike geodesics in $M$).

Once we have established this, we may choose $\Na$ and $\Nb$ to be representatives of set theoretic maximal elements $[N]_{\mathrm{max}}$ and $[N]'_{\mathrm{max}}$ to conclude that any two set theoretic maximal elements are equal (cf.\ Theorem \ref{theo:uniquemaximalextension}), which establishes that $\Nmax$ is indeed a maximal regular future g-boundary extension in the sense of Definition \ref{def:Maximal_regular_extension}.  

\begin{definition}\label{def:simarbitrary} Let $(\Na,\ia)$ and $(\Nb,\ib)$ be two regular future g-boundary extensions of a $C^2$ spacetime $(M,g)$ and let $p, q \in\Na\sqcup \Nb$. We say $ p\sim q $ if either
\begin{enumerate}
    \item $p\in \iota_a(M), q\in \iota_b(M)$ for some $a,b\in \{\alpha,\beta\}$ and $\iota_a^{-1}(p)=\iota_b^{-1}(q)$ or
    \item  $p\in N_a\setminus \iota_a(M), q\in  N_b\setminus \iota_b(M)$ for some $a,b\in \{\alpha,\beta\}$ and there exists a future directed timelike geodesic $\gamma:[0,1)\to M$ such that $\lim_{t\to 1^-}^{\Na}(\iota_a\circ\gamma_Y)(t)=p$ and $\lim_{t\to 1^-}^{\Nb}(\iota_b\circ\gamma_Y)(t)=q$ (note that this by definition requires both limits to exist). 
\end{enumerate}
\end{definition}

\begin{remark}
    \begin{enumerate}
    \item If $p$ and $q$ both lie in $\Na$ or both lie in $\Nb$, then $p\sim q$ iff $p=q$. 
        \item That this is indeed an equivalence relation (transitivity is not immediately obvious as we only demand the existence of a suitable $Y$ and this $Y$ may a priori depend on both $p$ and $q$) will follow from Lemma \ref{lem:one_geodesic_lim_eq_to_all_geodesic_lim_eq}.
        \item If one has $\Na \lesssim\Nb$ and w.l.o.g.\ $p\in \Na,q\in \Nb$, then $p\sim q$ according to Definition \ref{def:simarbitrary} if and only if $q=\psiab(p)$, i.e., if $p\sim q$ according to the definition in \eqref{eq:equivalencerelation}: If $q=\psiab(p)$, then for any future directed timelike geodesic $\gamma:[0,1)\to M$ with $\lim_{t\to 1^-}^{\Na}(\iota_a\circ\gamma_Y)(t)=p$ we have $\lim_{t\to 1^-}^{\Nb}(\iota_b\circ\gamma_Y)(t)= \lim_{t\to 1^-}^{\Nb}(\psiab\circ \iota_a\circ\gamma_Y)(t)=q$, so $p\sim q$ according to Definition \ref{def:simarbitrary}. On the other hand, if $p\sim q$ according to Definition \ref{def:simarbitrary}, then $\psiab$ being the continuous extension of $\ib\circ\ia^{-1}$ to all of $\Na$ implies $q=\psiab(p)$.
    \end{enumerate}
\end{remark}

\begin{lemma}\label{lem:one_geodesic_lim_eq_to_all_geodesic_lim_eq}Let $(\Na,\ia)$ and $(\Nb,\ib)$ be two regular future g-boundary extensions of a $C^2$ spacetime $(M,g)$. Let $p\in\Na\setminus \ia(M)$ and $q\in\Nb\setminus \ib(M)$. Then $ p\sim q $ if and only if for \emph{all} $Y\in T_tM$ with $\lim_{t\to b_Y^-}^{\Na}(\ia\circ\gamma_Y)(t)=p$ also $\lim_{t\to b_Y^-}^{\Nb}(\ib\circ\gamma_Y)(t)=q$ (and for \emph{all} $Y\in T_tM$ with $\lim_{t\to b_Y^-}^{\Nb}(\ib\circ\gamma_Y)(t)=q$ also $\lim_{t\to b_Y^-}^{\Na}(\ia\circ\gamma_Y)(t)=p$).
\end{lemma}
\begin{proof}
Fix $p$ and $q$ and $Y_0 \in T_tM$ 
such that $\lim_{t\to b_{Y_0}^-}^{\Na}(\ia\circ\gamma_{Y_0})(t)=p$ and $\lim_{t\to b_{Y_0}^-}^{\Nb}(\ib\circ\gamma_{Y_0})(t)=q$. Let $Y\in T_tM$ be some other vector with $\lim_{t\to b_{Y}^-}^{\Na}(\ia\circ\gamma_{Y})(t)=p$. We need to show that $\lim_{t\to b_{Y}^-}^{\Nb}(\ib\circ\gamma_{Y})(t)$ exists and equals $q$: Since $\Nb$ is a regular future g-boundary extension, the collection $\{O_{T_n,r_m}^{\Nb}\}_{n,m\in\mathbb{N}}$, where $T_n:=\dot{\gamma}_{Y_0}(1-\frac{1}{n})$ and $r_m=\frac{1}{m}$,  is a neighborhood basis of $q$. Since $\Na$ is a regular future g-boundary extension as well, $\{O_{T_n,r_m}^{\Na}\}_{n,m\in \mathbb{N}}$ is a neighborhood basis for $p$. 
    Since $\{O_{T_n,r_m}^{\Na}\}_{n,m\in \mathbb{N}}$ is a neighborhood basis for $p$ and $\ia\circ\gamma_Y\to p$ by assumption, for any $n\in \mathbb{N}$ we can find $0<t_{n,m}<b_Y$ such that $\ia\circ\gamma_Y(t)\in O_{T_n,r_m}^{\Na}$ for all $t\in (t_{n,m},b_Y)$. By the definitions of $O_{T_n,r_m}^{\Na}$ and $O_{T_n,r_m}^{M}$ and Remark \ref{rem:OM_def_and_open}, this implies $\gamma_Y(t)\in O_{T_n,r_m}^{M}$ for all $t\in (t_{n,m},b_Y)$. Hence, again appealing to the definitions and Remark \ref{rem:OM_def_and_open}, we obtain $\ib\circ \gamma_Y(t)\in O_{T_n,r_m}^{\Nb}$ for all $t\in (t_{n,m},b_Y)$. Since this works for any $n,m$ and $\{O_{T_n,r_m}^{\Nb}\}_{n,m\in\mathbb{N}}$ is a neighborhood basis for $q$ we get $\lim_{t\to b_{Y}^-}^{\Nb}(\ib\circ\gamma_{Y})(t)=q$.
   \end{proof}

     Note that it was essential for the above proof that we could choose the same $T_n=\dot{\gamma}_{Y_0}(1-\frac{1}{n}), r_m$ for the neighborhood bases in $\Na$ and in $\Nb$ by the second condition in Definition \ref{def:regular-g-boundary-ext} because we already had one geodesic $\gamma_{Y_0}$ with the right limiting behavior in $\Na$ and $\Nb$. We will encounter this again when showing that $\pi$ is an open map.

So, $\sim$ from Definition \ref{def:simarbitrary} is indeed an equivalence relation and we may define the quotient space \begin{equation}\label{eq:defTildeN}
    \Tilde{N}\coloneqq (\Na\sqcup \Nb)/\sim.
\end{equation} As in Section \ref{subsec:settheoreticmax} we equip $\Tilde{N}$ with the quotient topology $\tau_q$. We proceed by showing that also in this case the quotient map $\Tilde{\pi}: \Na\sqcup \Nb\to \tilde{N} $ is open provided $(M,g)$ does not contain any intertwined timelike geodesics.

\begin{lemma}\label{lem:pitilde_open}
Let $(\Na,\ia)$ and $(\Nb,\ib)$ be two regular future g-boundary extensions of a strongly causal $C^2$ spacetime $(M,g)$. If no two timelike geodesics $\gamma_1,\gamma_2:[0,1)\to M$ with $\iota_{\alpha}\circ \gamma_1 $ converging to a $p_1\in  N_{\alpha}\setminus \iota_{\alpha}(\Na)$ and $\iota_{\beta}\circ \gamma_2 $ converging to a $p_2\in N_{\beta}\setminus \iota_{\beta}(\Nb)$ are intertwined,
then the projection map $\Tilde{\pi}:\Na\sqcup \Nb\rightarrow \Tilde{N}$ is open. 
\end{lemma}

  \noindent \textit{Proof.}  By definition of the quotient topology in $(\Na\sqcup \Nb)/\sim$, the projection map $\Tilde{\pi}$ is open if and only if for $a,b\in \{\alpha,\beta\}$ we have that for any open $U\subset N_a$ the image $\Tilde{\pi}(U)$ is open, i.e., the preimage $\Tilde{\pi}^{-1}(\Tilde{\pi}(U))=U\cup (\Tilde{\pi}|_{N_b})^{-1}(\Tilde{\pi}(U))$ is open in $\Na \sqcup\Nb$.  Hence, restricting ourselves to the exemplary case of  $a=\alpha$, $b=\beta$ for simplicity, it is sufficient to show $(\Tilde{\pi}|_{\Nb})^{-1}(\Tilde{\pi}(U))$ is open for any open $U\subset \Na$. 
    So let $U\subset \Na$ be open. We show that for any $q\in(\Tilde{\pi}|_{\Nb})^{-1}(\Tilde{\pi}(U))$ there exists an open neighborhood $V\subset \Nb$ of $q$ with $\Tilde{\pi}(V)\subset \Tilde{\pi}(U)$, i.e. satisfying $V\subset (\Tilde{\pi}|_{\Nb})^{-1}(\Tilde{\pi}(U))$.

    If $q\in \ib(M)$, then for any open neighborhood $V'\subset M$ of $\ib^{-1}(q)$ in $M$ we then have that $ V:=\ib(\ia^{-1}(U)\cap V') $ is open in $\Nb$ (note that $\ib$ is an open map by assumption), contains $q$ and clearly satisfies $\Tilde{\pi}(V) \subset \Tilde{\pi}(U)$ by the definition of the equivalence relation.
    
   The more interesting case is $q\notin \ib(M)$. This implies $\Tilde{\pi}(q)=\Tilde{\pi}(p)$ for a unique $p\in U\setminus  \ia(M)$ and that there exists $Y\in T_tM$ with  $\lim_{t\to 1^-}^{\Na}(\ia\circ\gamma_Y)(t)=p$ and  $\lim_{t\to 1^-}^{\Nb}(\ib\circ\gamma_Y)(t)=q$. Let us again denote  $T_n:=\dot{\gamma}_{Y}(1-\frac{1}{n})$ and $r_m:=\frac{1}{m}$. We will show that there exist $n_0,m_0\in \mathbb{N}$ for which we may take $V=O^{\Nb}_{T_{n_0},r_{m_0}}$, i.e., that $\Tilde{\pi}(O^{\Nb}_{T_{n_0},r_{m_0}})\subset \Tilde{\pi}(U)$.

    Since $\{O_{T_n,r_m}^{\Na}\}_{n,m\in \mathbb{N}}$ is a neighborhood basis at $p$ (remembering that $\Na$ is a regular future g-boundary extension and $p=\lim_{t\to 1^-}\ia\circ\gamma_Y(t)$) and $U$ is an open neighborhood of $p$, we must have $O_{T_{n_0},r_{m_0}}^{\Na}\subset U$ for some $n_0,m_0$. Since $\Na$ is a topological manifold, we can further w.l.o.g.\ assume that $\overline{O_{T_{n_0},r_{m_0}}^{\Na}}$ is compact and also contained in $U$. Fix these $n_0,m_0\in \mathbb{N}$ and assume $\Tilde{\pi}(O^{\Nb}_{T_{n_0},r_{m_0}})\not\subset \Tilde{\pi}(U)$. Then there exists $q_0\in O^{\Nb}_{T_{n_0},r_{m_0}}$ with $\Tilde{\pi}(q_0)\notin \Tilde{\pi}(U)$.

   We now distinguish two cases: Either $q_0$ is contained in $ \ib(M)$ or $q_0\notin \ib(M)$. In the first case 
   $\ib^{-1}(q_0)\in O_{T_{n_0},r_{m_0}}^{M}$ implying $\ia(\ib^{-1}(q_0))\in O_{T_{n_0},r_{m_0}}^{\Na}\subset  U$. This contradicts $\Tilde{\pi}(q_0)\notin \Tilde{\pi}(U)$. 
   
     \begin{wrapfigure}{r}{0.4 \textwidth}
  	\centering \includegraphics[width=0.38\textwidth]{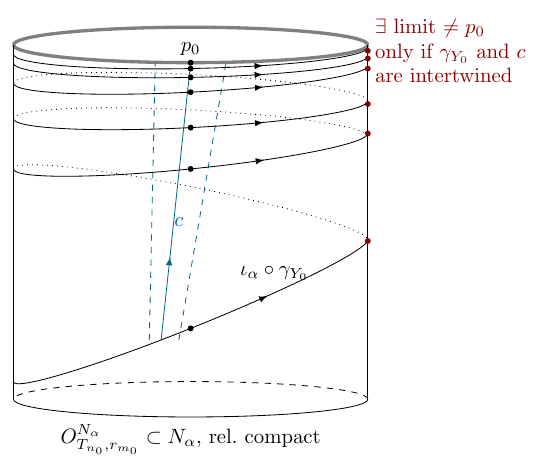}
  	\caption{Illustration of the last part of the proof based loosely on the situation in the Milne spacetime, which does not admit a unique maximal extension: If there were two different limits $p_0:=\lim \ia \circ \gamma_{Y_0}(t_k)\neq \lim \ia\circ \gamma_{Y_0}(t'_k)$, then any future directed timelike geodesic $c$ in $M$ with $\ia\circ c$ terminating in $p_0$ would be intertwined with the original $\gamma_{Y_{0}}$.}
  	\label{fig:2}
  \end{wrapfigure} So $q_0\notin \ib(M)$. Since $q_0\in \Nb\setminus  \ib(M)$, there exists $Y_0\in T_tM$ with $b_{Y_0}= 1$ and $$\lim_{t\to 1^-} \ib\circ\gamma_{Y_0}(t)=q_0.$$ 
  Because $\lim_{t\to 1^-} \ib\circ\gamma_{Y_0}(t)=q_0$ and $q_0\in O_{T_{n_0},r_{m_0}}^{\Nb}$ is open, there exists $t_0$ such that $\ib\circ\gamma_{Y_0}([t_0,1))\subset O_{T_{n_0},r_{m_0}}^{\Nb}$, hence $\gamma_{Y_0}([t_0,1))\subset O_{T_{n_0},r_{m_0}}^{M}$ and thus 
  $\ia\circ\gamma_{Y_0}([t_0,1))\subset O_{T_{n_0},r_{m_0}}^{\Na}$. 
  Now note that \emph{if} $$\lim_{t\to 1^-} \ia\circ\gamma_{Y_{0}}(t)=:p_{0}$$ exists in $\Na$ then we will have  $\Tilde{\pi}(p_{0})=\Tilde{\pi}(q_{0})$ by definition and continuity of $\Tilde{\pi}$ (remembering that $q_{0}=\lim_{t\to 1^{-}}(\ib\circ\gamma_{Y_{0}})(t)$). Further, since we chose $n_0,m_0$ such that $\overline{O_{T_{n_0},r_{m_0}}^{\Na}}\subset U$, the limit $p_{0}$ must be in $U$. So together this would contradict $\Tilde{\pi}(q_{0})\notin \Tilde{\pi}(U)$.

   It thus remains to show that this limit exists. By relative compactness of $O^{\Na}_{T_{n_0},r_{m_0}}$ there always exists a sequence $t_k\to 1$ such that $\ia \circ \gamma_{Y_{0}}(t_k)$ converges. Let's denote this limit by $p_{0}$. We will exploit the fact that $(M,g)$ does not contain any intertwined timelike geodesics to argue that actually $\lim_{t\to 1^-} (\ia\circ\gamma_{Y_{0}})(t)=p_{0}$.
    Since    $p_{0}\in \Na\setminus\ia(M)$, there must exist some timelike geodesic $c : [0,1)\to M$  such that $p_{0}=\lim_{t\to 1^-} (\ia\circ c)(t)$. Since $c$ and   $\gamma_{Y_{0}}$ cannot be intertwined by assumption they either satisfy point (i) in Definition \ref{defintertwined}, in which case Lemma \ref{lem:propertiesthickenings} applies and we can conclude that for any $n,m\in \mathbb{N}$ there exists $s'\equiv s'(n,m)$ such that $(\ia\circ \gamma_{Y_{0}})([s',1))\in O^{\Na}_{\dot{c}(1-\frac{1}{n}),\frac{1}{m}}$, implying that $\lim_{t\to 1^-} (\ia\circ\gamma_{Y_{0}})(t)=\lim_{t\to 1^-} \ia(c(t))=p_{0}$ by the neighborhood-basis property of $\{O^{\Na}_{\dot{c}(1-\frac{1}{n}),\frac{1}{m}}\}_{n,m \in \mathbb{N}}$ (and Hausdorffness of $\Na$). Or they satisfy point (ii) in Definition \ref{defintertwined}, that is there exist $s_1,s_2\in (0,1)$ and $r,\r>0$ such that $O_{\dot{\gamma}_{Y_{0}}(s_1),r}^{M} \cap O_{\dot{c}(s_2),\r}^{M}=\emptyset $. But this is impossible because for any $s_1,s_2\in (0,1)$ and $r,\r>0$ we have $ \gamma_{Y_{0}}(t_k)\in O^M_{\dot{\gamma}_{Y_{0}}(s_1),r} \cap O_{\dot{c}(s_2),\r}^{M}$ for all large enough $k$: On the one hand, for any $s_1\in(0,1)$ there clearly exists $K$ such that $t_k\geq s_1$ for all $k\geq K$ and then $\gamma_{Y_{0}}(t_k)\in O^M_{\dot{\gamma}_{Y_{0}}(s_1),r}$ for all $r>0$. On the other hand, for any $s_2\in (0,1),\r>0$ the set $O^{\Na}_{\dot{c}(s_2),\r}$ is an open set (as $\Na$ is a regular future g-boundary extension), contains $p_{0}=\lim_{t\to 1}(\ia\circ c)(t)$ and 
   $\ia \circ \gamma_{Y_{0}}(t_k)\to p_{0}$, so there also exists $K$ such that $\ia \circ\gamma_{Y_{0}}(t_k)\in O^{\Na}_{\dot{c}(s_2),\r}$ for all $k\geq K$. \qed \\

As in Section \ref{subsec:settheoreticmax} we define a map $\Tilde{\iota}: M \rightarrow \Tilde{N}$ via
\begin{equation}\label{eq:iota_tilde}
   \Tilde{\iota}(p):= \Tilde{\pi}(\ia(p))=\Tilde{\pi}(\ib(p)).
\end{equation}
This map is well-defined and, since $\Tilde{\pi}$ is an open map, a (topological) embedding onto the open set $\Tilde{\iota}(M)\subset \Tilde{N}$. Further $\overline{\Tilde{\iota}(M)}^{\tau_q}=\Tilde{N}$ since $\overline{\ia(M)}^{\Na}\sqcup \overline{\ib(M)}^{\Nb}=\Na\sqcup \Nb$ and $\Tilde{\pi} $ is continuous, so $((\Tilde{N},\tau_q),\Tilde{\iota})$ is a candidate for future boundary extension of $(M,g)$. We may now appeal to Lemma \ref{lem:Nmax_regular_candidate} and Remark \ref{rem:regular_ok_whenever_quotient_compatible_and_pi_open} to conclude that in fact

\begin{lemma}\label{lem:Ntilde_regular_candidate}
Under the assumptions of Lemma \ref{lem:pitilde_open}     $((\Tilde{N}, \tau_q),\Tilde{\iota})$ is a candidate for a regular future g-boundary extension of $(M,g)$.
 \end{lemma}

Next it is shown that, provided there are no intertwined timelike geodesics in $M$, $(\Tilde{N}, \tau_q)$ is Hausdorff. 

\begin{lemma}\label{lem:NtildeHausdorff}
    Let $(\Na,\ia)$ and $(\Nb,\ib)$ be two regular future g-boundary extensions of a strongly causal $C^2$ spacetime $(M,g)$, $\Tilde{N}\coloneqq (\Na\sqcup \Nb)/\sim$, $\Tilde{\iota}$ as in expression \eqref{eq:iota_tilde} and $\tau_{q}$ the quotient topology on $\Tilde{N}$. If no two timelike geodesics $\gamma_1,\gamma_2:[0,1)\to M$ with $\iota_{\alpha}\circ \gamma_1 $ converging to a $p_1\in  N_{\alpha}\setminus \iota_{\alpha}(\Na)$ and $\iota_{\beta}\circ \gamma_2 $ converging to a $p_2\in N_{\beta}\setminus \iota_{\beta}(\Nb)$ are intertwined,
    then $\Tilde{N}$ is a topological Hausdorff space.
\end{lemma}

\begin{proof}
    Consider two distinct points $p,q\in \Tilde{N}$. We will separate two cases: Either $p,q\in \Tilde{\pi}(N_a)$ for some $a\in \{\alpha,\beta\}$ or $p\in \Tilde{\pi}(N_a)\setminus \Tilde{\pi}(N_b)$ and $q\in \Tilde{\pi}(N_b)\setminus\Tilde{\pi}(N_a)$ for some $a,b\in \{\alpha,\beta\}$ with $a\neq b$. So, let $p,q\in \Tilde{\pi}(N_a)$ and let $p_a,q_a\in N_a$ be the unique points such that $p=\Tilde{\pi}(p_a)$ and $q=\Tilde{\pi}(q_a)$ (noting that $\Tilde{\pi}\big|_{N_a}$ is injective). Hausdorffness of $N_a$ implies that there exist disjoint neighborhoods $U,V\subset N_a$ of $p_a$ and $q_a$ respectively and hence, by openness of $\Tilde{\pi}$ and injectivity of $\Tilde{\pi}|_{N_a}$, $\Tilde{\pi}(U)$ and $\Tilde{\pi}(V)$ are disjoint open neighborhoods of $p$ and $q$ respectively. 

    It remains to show that there exist disjoint neighborhoods of $p,q$ when $p\in \Tilde{\pi}(\Na)\setminus\Tilde{\pi}(\Nb)$ and $q\in \Tilde{\pi}(\Nb)\setminus\Tilde{\pi}(\Na)$ (or vice versa). Let $\pa\in\Na$ and $\qb\in\Nb$ be the unique points such that $p=\Tilde{\pi}(\pa)$ and $q=\Tilde{\pi}(\qb)$. This implies that $\qb\in \Nb\setminus \ib(M)$ and $\pa \in \Na\setminus \ia(M)$ (otherwise, $q\in\Tilde{\pi}(\Na)$ or $p\in \Tilde{\pi}(\Nb)$). Hence, there exist timelike geodesics $\gamma_{1}: [0,1)\rightarrow M$ and $\gamma_{2} : [0,1)\rightarrow M$ with $\lim_{s\to 1}(\ia\circ\gamma_{1})(s)= \pa$ and $\lim_{s\to 1}(\ib\circ\gamma_{2})(s)= \qb$, which, by assumption, are not intertwined. In other words, $\gamma_{1}$ and $\gamma_{2}$ satisfy either condition (i) or (ii) of Definition \ref{defintertwined} which we will now discuss separately.

    In the first place, suppose that for any radii $r>0$, $\rho> 0$ there exist $s_{1}\in (0,1)$, $s_{2}\in\left(0,1\right)$ such that $\ga(\left[s_{1},1)\right)\subset O^M_{\dot{\gamma}_2(0),\rho}$
    and $\gb(\left[s_{2},1)\right)\subset O^M_{\dot{\gamma}_1(0),r}$. Moreover, let $\{O^{\Na}_{X_n,r_m}\}_{n,m\in\mathbb{N}}$ and $\{O^{\Nb}_{Y_n,\rho_m}\}_{n,m\in\mathbb{N}}$ be the associated neighborhood basis of $\pa$ and $\qb$ with $X_{n}=\dot\gamma_{1}(1-\frac{1}{n})$, $r_{m}=1/m$, $Y_{n}=\dot\gamma_{2}(1-\frac{1}{n})$ and $\rho_{m}=1/m$. Then, by Lemma \ref{lem:propertiesthickenings}, it holds that there exist $s'_{1}\in(0,1), s'_{2}\in(0,1)$ such that $\ga(\left[s'_{1},1)\right)\subset O^{M}_{Y_n,\rho_m}$ and $\gb(\left[s'_{2},1)\right)\subset O^{M}_{X_n,r_m}$.
    Thus, for all $n,m\in \mathbb{N}$, there exist $s'_{1}(n,m)\in (0,1), s'_{2}(n,m)\in(0,1)$ such that $(\ib\circ\ga)(s)\subset O^{\Nb}_{Y_n,\rho_m}\cap\ib(M)$  for $s\in[s'_{1},1)$
    and $(\ia\circ\gb)(s)\subset O^{\Na}_{X_n,r_m}\cap\ia(M)$ for $s\in\left[s'_{2},1\right)$. 
    Since $O^{\Na}_{X_n,r_m}$ and $O^{\Nb}_{Y_n,\rho_m}$ are neighborhood bases of $\pa$ and $\qb$ respectively (and by Hausdorffness of $\Na$ and $\Nb$), this implies that $\lim_{s\to 1^{-}}(\ia\circ\gamma_{2})(s)= \pa$ and $\lim_{s\to 1^{-}}(\ib\circ\gamma_{1})(s)= \qb$. So in particular we have $\lim_{s\to 1^{-}}(\ia\circ\gamma_{2})(s)= \pa$ but we originally chose $\gamma_2$ such that $\lim_{s\to 1^{-}}(\ib\circ\gamma_{2})(s)= \qb$, so by definition of the equivalence relation, we have that $\Tilde{\pi}(\pa)=\Tilde{\pi}(\qb)$, which contradicts our initial assumption that $p\neq q$.
    
    Finally, consider the case that there exist $s_{1}\in (0,1)$, $s_{2}\in\left(0,1\right)$ and radii $r,\r >0$ such that $O^M_{\dot{\gamma}_{1}(s_1),r} \cap O^M_{\dot{\gamma}_{2}(s_2),\r}=\emptyset$. As $\Tilde{\iota}$ is injective, $\emptyset=\Tilde{\iota}(O^{M}_{\dot \gamma_{1}(s_{1}),r}\cap O^{M}_{\dot \gamma_{2}(s_{2}),\r})=\Tilde{\iota}(O^{M}_{\dot \gamma_{1}(s_{1}),r})\cap \Tilde{\iota}(O^{M}_{\dot \gamma_{2}(s_{2}),\r})$. Furthermore, since $((\Tilde{N},\tau_q),\Tilde{\iota})$ is a candidate for a future boundary extension, Remark \ref{rem:OM_def_and_open} implies that $\Tilde{\iota}(O^{M}_{\dot \gamma_{1}(s_{1}),r})=O^{\Tilde{N}}_{\dot \gamma_{1}(s_{1}),r}\cap\Tilde{\iota}(M)$ and $\Tilde{\iota}(O^{M}_{\dot \gamma_{2}(s_{2}),\rho})=O^{\Tilde{N}}_{\dot \gamma_{2}(s_{2}),\r}\cap\Tilde{\iota}(M)$. This gives us that:
    
    \begin{equation}\label{eq:Hausdorff1}
    (O^{\Tilde{N}}_{\dot \gamma_{1}(s_{1}),r}\cap\Tilde{\iota}(M)) \cap (O^{\Tilde{N}}_{\dot \gamma_{2}(s_{2}),\rho}\cap \Tilde{\iota}(M))=\emptyset
    \end{equation}
    
    It remains to show that \eqref{eq:Hausdorff1} actually implies that also $O^{\Tilde{N}}_{\dot \gamma_{1}(s_{1}),r}\cap O^{\Tilde{N}}_{\dot \gamma_{2}(s_{2}),\r} \cap \overline{\Tilde{\iota}(M)}=O^{\Tilde{N}}_{\dot \gamma_{1}(s_{1}),r}\cap O^{\Tilde{N}}_{\dot \gamma_{2}(s_{2}),\r}=\emptyset $. This follows easily by contradiction. Assume that there exists a point $r\in \overline{\Tilde{\iota}(M)}$ with $r\in O^{\Tilde{N}}_{\dot \gamma_{1}(s_{1}),r} \cap O^{\Tilde{N}}_{\dot \gamma_{2}(s_{2}),\rho}$. Then, openness of $O^{\Tilde{N}}_{\dot \gamma_{1}(s_{1}),r} \cap O^{\Tilde{N}}_{\dot \gamma_{2}(s_{2}),\rho}$ in $\Tilde{N}$ (which we already established with Lemma \ref{lem:Ntilde_regular_candidate}) together with the definition of the closure and a standard topological argument implies that $O^{\Tilde{N}}_{\dot \gamma_{1}(s_{1}),r} \cap O^{\Tilde{N}}_{\dot \gamma_{2}(s_{2}),\rho}\cap \Tilde{\iota}(M) \neq \emptyset $, a contradiction to expression \eqref{eq:Hausdorff1}. So, $p\in O^{\Tilde{N}}_{\dot \gamma_{1},r}$, $q\in O^{\Tilde{N}}_{\dot \gamma_{2},\rho}$ and $O^{\Tilde{N}}_{\dot \gamma_{1}(s_{1}),r} \cap O^{\Tilde{N}}_{\dot \gamma_{2}(s_{2}),\rho}=\emptyset$.
\end{proof}

Lastly, coordinate charts can be defined on $\Tilde{N}$. This process is again analogous to Section \ref{subsec:settheoreticmax}.

\begin{lemma}\label{lem:chartsNtilde}
 Let $(\Na,\ia)$ and $(\Nb,\ib)$ be two regular future g-boundary extensions of a strongly causal $C^2$ spacetime $(M,g)$, $\Tilde{N}\coloneqq (\Na\sqcup \Nb)/\sim$ and $p\in \Tilde{N}$. Then there exists an open neighborhood $U$ of $p\in \Tilde{N}$ and a homeomorphism $x:U\to x(U)\subset [0,\infty)\times \mathbb{R}^{d-1}$ onto an open subset in the half space. In particular, $\Tilde{N} $ is a topological manifold with boundary. 
\end{lemma}
\begin{proof}
     Let $p\in \Tilde{N}$, w.l.o.g.\ $p\in \Tilde{\pi}(\Na)$, and choose $p_{\alpha}\in \Na$ such that $p=\Tilde{\pi}(p_{\alpha})$. Let $(U_{\alpha},x_{\alpha})$ be a coordinate chart around $p_{\alpha}$ in $\Na$. 
     As $\Tilde{\pi}$ is an open map, $\Tilde{\pi}(U_{\alpha})$ is an open neighborhood of $p$ in $\Tilde{N}$. Then, on $\Tilde{N}$ we define the map $x : \Tilde{\pi}(U_{\alpha})\rightarrow [0,\infty)\times \mathbb{R}^{d-1}, p\mapsto x_{\alpha}((\Tilde{\pi}\big|_{\Na})^{-1}(p))$. As the composition of injective, continuous and open maps this map is a homeomorphism onto the open set $x(\Tilde{\pi}(U_{\alpha}))\subset [0,\infty)\times \mathbb{R}^{d-1}$. 
\end{proof}

Now we can collect all the results we have shown in order to prove our second main 
Theorem:

\begin{theorem}\label{theo:uniquemaximalextension}
    Let $(M,g)$ be a strongly causal $C^2$ spacetime. If $(M,g)$ is regular future g-boundary extendible and does not contain any intertwined future directed timelike geodesics, then there exists a maximal  regular future g-boundary extension in the sense of Definition \ref{def:Maximal_regular_extension}. 
\end{theorem}
\begin{proof}
Let  $[N]_{\mathrm{max}}$ and $[N]'_{\mathrm{max}}$ be two set theoretic maximal elements. Choose representatives $(\Na,\ia)$ and $(\Nb,\ib)$. Let $\Tilde{N}= (\Na\sqcup \Nb)/\sim$. Collecting the previous results of this section, if $(M,g)$ does not contain any intertwined future directed timelike geodesics, $\Tilde{N}$ is a topological manifold with boundary, any point in $\pr\Tilde{N}$ is the limit point of a geodesic in $M$ and its 
collection of timelike thickenings defines a topology on $\Tilde{N}$ which agrees with the quotient topology. Finally, condition 2. of Definition \ref{def:regular-g-boundary-ext} is automatically satisfied as $\Tilde{\pi}$ is an open map. Hence, $\Tilde{N}$ is a regular 
future g-boundary extension.  We may thus consider $[\Tilde{N}]$. Since $\Tilde{\pi}:\Na\sqcup \Nb \to \tilde{N}$ is an open map (by Lemma \ref{lem:pitilde_open}, again noting that $(M,g)$ does not contain any intertwined future directed timelike geodesics), the composition $\Na\hookrightarrow \Na\sqcup \Nb  \overset{\Tilde{\pi}}{\to} \Tilde{N}$ is an embedding and we have $[N]_{\mathrm{max}}=[\Na] \lesssim [\Tilde{N}]$, so set-theoretic maximality guarantees 
$[N]_{\mathrm{max}} = [\Tilde{N}]$. Since we can do the same argument for the composition $\Nb\hookrightarrow \Na\sqcup \Nb  \overset{\Tilde{\pi}}{\to} \Tilde{N}$, we have $[N]_{\mathrm{max}} = [\Tilde{N}]=[N]_{\mathrm{max}}' $. Accordingly, there must be a unique set theoretic maximal element $[N]_{\mathrm{max}}$ and hence by Remark \ref{rem:maxextension} any representative for $[N]_{\mathrm{max}}$ is a maximal regular future g-boundary extension.
\end{proof}

\begin{remark}\label{rem:dezorn}
    Sections \ref{subsec:settheoreticmax} and \ref{subsec:actualmax} are actually more independent of each other than they might seem at first reading. This is important as Corollary \ref{cor:existence_of_set_theoretic_max} relies heavily on Zorn's Lemma while Theorem \ref{theo:uniquemaximalextension} doesn't. Moreover, if we assumed that $(M,g)$ contains no intertwined timelike geodesics from the beginning, our proof would not need to resort to Zorn's Lemma (and, in particular, we would be able to not only prove the existence of a unique maximal extension, but also to construct it). In the following, let $\Na$ and $\Nb$ be two arbitrary regular future g-boundary extensions of a strongly causal spacetime $(M,g)$:
    \begin{itemize}
        \item In the proof of Theorem \ref{theo:uniquemaximalextension} we actually show that if $(M,g)$ contains no intertwined timelike geodesics, gluing together $\Na$ and $\Nb$ (and identifying points appropriately as in Definition \ref{def:simarbitrary}) yields a 'larger' regular future g-boundary extension $\Tilde{N}$.
        \item However, the previous conclusion can also be generalized to an arbitrarily large number of regular future g-boundary extension, assuming again that $(M,g)$ has no intertwined timelike geodesics. Let $\mathcal{I}$ be the set of regular future g-boundary extensions of $(M,g)$. Then, $\Tilde{N}\coloneqq(\bigsqcup_{\alpha\in \mathcal{A}} \Na)/\sim$ is a regular future g-boundary extension (by the proof of Theorem \ref{theo:uniquemaximalextension} and using that $\Tilde{N}$ is second countable even for uncountable unions as it is a candidate for a regular future g-boundary extension). As any $\Na\in\mathcal{I}$ can be embedded in $\Tilde{N}$, it is clear that this is the unique maximal regular future g-boundary extension.
        \item The 'dezornified' version of Theorem \ref{theo:uniquemaximalextension} is similar to Sbierski's dezornification \cite{SbierskiMGHD} of the proof of the existence of a unique maximal globally hyperbolic development of a given initial data set by Choquet-Bruhat and Geroch. In the first place, the statement that $\Tilde{N}= (\Na\sqcup \Nb)/\sim$ is a regular future g-boundary extension is similar to Theorem 2.7 in \cite{SbierskiMGHD}\footnote{This theorem states that for any two globally hyperbolic developments of the same initial data, there exists a 'larger' globally hyperbolic development in which they both isometrically embed and which is constructed by gluing them together along the maximal common globally hyperbolic development.}. Secondly, the strategy of using the previous result to glue all regular future g-boundary extensions together in order to construct the maximal regular future g-boundary extension is very similar to Theorem 2.8 in \cite{SbierskiMGHD}, which states the existence of a unique maximal globally hyperbolic development\footnote{Note that in the proofs of Theorem 2.7 and 2.8 in \cite{SbierskiMGHD} a very important step is to identify points lying in the 'common globally hyperbolic development' of two arbitrary globally hyperbolic developments. This identification is analogous to our Definition \ref{def:simarbitrary}, where the limit points (in the different $\Na$'s) of the same inextendible timelike geodesic (in M) are identified.}.
    \end{itemize}
\end{remark}

Let us remark that we can also formulate an ''if and only if'' version of Theorem \ref{theo:uniquemaximalextension} as follows: A regular future g-boundary extendible strongly causal $C^2$ spacetime $(M,g)$ has a maximal future g-boundary extension if and only if $(M,g)$ does not admit any intertwined timelike geodesics $\gamma_1$ and $\gamma_2$ such that $\iota_1\circ \gamma_1$ acquires an endpoint in some regular future g-boundary extension $(N_1,\iota_1)$  and $ \iota_2\circ\gamma_2$ acquires an endpoint in some other regular future g-boundary extension $(N_2,\iota_2)$. Clearly this is sufficient for the proof of Theorem \ref{theo:uniquemaximalextension} to go through. For the ''only if'' part we note first that 
if $\Nmax$ is a maximal regular future g-boundary extension, then there cannot exist any intertwined timelike geodesics $\gamma_1$ and $\gamma_2$ in $M$ such that $\ima\circ \gamma_1$ and $\ima\circ \gamma_2$ have future endpoints in $\Nmax$: Assume that $\gamma_1$ and $\gamma_2$ are two geodesics in $M$ such that $\ima\circ \gamma_1$ and $\ima\circ \gamma_2$ have future endpoints $p_1$ respectively $p_2$ in $\Nmax$, then either
\begin{itemize}
    \item $p_1=p_2$, in which case $\gamma_1$ and $\gamma_2 $ will satisfy  condition (i) in Definition \ref{defintertwined} because $O^{\Nmax}_{\dot{\gamma}_1(0),r}$ is an open neighborhood of $p_2$ and vice versa. 
    \item $p_1\neq p_2$, in which case Hausdorffness of $\Nmax$ together with $\{O^N_{\dot{\gamma}_1(1-\frac{1}{n}),\frac{1}{m}}:n,m\in\mathbb{N}\} $ and $\{O^N_{\dot{\gamma}_2(1-\frac{1}{n}),\frac{1}{m}}:n,m\in\mathbb{N}\} $ being neighborhood bases at $p_1$ resp.\ $p_2$ allows us to find disjoint  $O^N_{\dot{\gamma}_1(1-\frac{1}{n_1}),\frac{1}{m_1}}$ and   $O^N_{\dot{\gamma}_2(1-\frac{1}{n_2}),\frac{1}{m_2}}$. 
    This immediately implies that $\gamma_1$ and $\gamma_2 $ will satisfy  condition (ii) in Definition \ref{defintertwined}.
\end{itemize}
So in both cases $\gamma_1$ and $\gamma_2$ are not intertwined. Since any regular future g-boundary extension embeds into $\Nmax$ the existence of a maximal regular future g-boundary extension further implies that $M$ cannot have any intertwined timelike geodesics $\gamma_1$ and $\gamma_2$ in $M$ such that $\iota_1\circ \gamma_1$ acquires an endpoint in some regular future g-boundary extension $(N_1,\iota_1)$  and $ \iota_2\circ\gamma_2$ acquires an endpoint in some other regular future g-boundary extension $(N_2,\iota_2)$. All of this is in line with the corresponding converse statement for  conformal boundary extension  in \cite[Thm.\ 4.5]{Chrusciel}.  However, it is at this point unclear to us if this would already imply somehow that 
$(M,g)$ cannot contain any intertwined future directed timelike geodesics at all.

\section{Discussion}\label{sec:discussion}

We already discussed some of the limitations of our approach (such as only getting regular future g-boundary extensions and not necessarily spacetime extensions and requiring quite a bit of 'hidden' regularity), 
so we would like to end with several possibilities and open questions for extending our work and potential applications thereof. Most of these have been mentioned throughout the paper already but we will collect them here and give a little more detail.

First, let us note that even though our results are only about the existence of a maximal future g-boundary extension this may have some consequences for $C^2$ spacetime extensions as well because of the compatibility results in Section \ref{sec:compatibility}. For instance, if for a given $C^2$ spacetime which is globally hyperbolic, past timelike geodesically complete and contains no intertwined timelike geodesics one could show that the maximal future g-boundary extension has non-compact and connected boundary, then maximality and invariance of domain imply that no $C^2$ spacetime extension can have compact $\pr^+\iota(M)$. \\[-0.5em]

Since $C^2$ spacetime extensions are anyways rather nice and well understood an important follow up question would be how far one can lower the regularity of a spacetime extension $(\mext,\gext)$ in Section \ref{sec:compatibility} while retaining its conclusions. That $\gext \in C^{1,1}$ is sufficient should be very straightforward to check. If $\gext \in C^1$, openness of all $O^N_{X,r}$ is still expected to be unproblematic and, with some more work, also the neighborhood property of $\{O^N_{\dot{\gamma}(1-\frac{1}{n}),\frac{1}{m}}:n,m\in \mathbb{N}\} $ should go through. However Sbierski's result guaranteeing the existence of timelike geodesics reaching any point in $\pr^+\iota(M)$ (cf. \cite[Lem.\ 3.1]{Sbierski2022} resp.\ Lemma \ref{lem:geodesic_to_boundary}) does at least on first read appear to really rely on facts about geodesics which fail if $\gext$ is merely $C^1$. 
 
Similarly, one may ask if it is really necessary to assume global hyperbolicity of $M$ or that $\pr^-\iota(M)=\emptyset$ for compatibility. Regarding global hyperbolicity we note 
 that if we keep $\pr^-\iota(M)=\emptyset$ even without global hyperbolicity we still get a future boundary extension: Since \cite[Thm.\ 2.6]{GallowayLing} establishes that then $\pr^+\iota(M)$ is always an achronal topological hypersurface and a standard argument then produces suitable charts near this hypersurface \cite[Prop.\ 14.25]{Oneill}. The arguments in Lemma \ref{lem:compatibility:regularityOK} should also go through without global hyperbolicity. However global hyperbolicity cannot be dropped from Lemma \ref{lem:geodesic_to_boundary} (i.e.\ \cite[Lem.\ 3.1]{Sbierski2022}), cf.\ \cite[Rem.\ 3.4, 3.]{Sbierski2022}, so we will not obtain a \emph{regular} future g-boundary extension without it.

Regarding the $\pr^-\iota(M)=\emptyset$ assumption one has to wonder whether $\pr^+\iota(M)\cap \pr^-\iota(M)= \emptyset$ (which excludes the obvious counterexample of an open cylinder where the ends are identified in the extension and $\iota(M) \cup \pr^+\iota(M)=\mext$, so is not a manifold with boundary with the subspace topology)  would be enough for the conclusions in Section \ref{sec:compatibility}. Almost all arguments work out nicely except for Lemma \ref{lem:no_further_intersections} and Lemma \ref{lem:compatibility:regularityOK}: Since we no longer have global achronality of $\pr^+\iota(M)$ the equality $O^N_{X,r}=O^{\mathrm{ext},\varepsilon }_{X,r}\cap N$ from the proof of Lemma \ref{lem:compatibility:regularityOK} might fail globally. Indeed one can imagine examples of ''accumulating'' boundaries where $O^{\mathrm{ext},\varepsilon}_{X,r}\cap N\neq O^N_{X,r}$ for all $\varepsilon>0$. 
Another approach could be to  abandon the subspace topology induced from $\mext$ and try to just equip $\iota(M)\cup \pr^+\iota(M)$ with the induced topology from the $O_{X,r}^N$ (or even from all $\iota(U)$'s and $\{O^N_{\dot{\gamma}(1-\frac{1}{n}),\frac{1}{m}}:n,m\in \mathbb{N}\}$'s). While this solves some issues (like no longer needing to prove significant parts of Lemma \ref{lem:compatibility:regularityOK}), it invariably introduces others (Hausdorffness, manifold structure, etc.). This is again reminiscent of similar issues in picking a suitable topology in the various boundary constructions such as of course the g-boundary of Geroch \cite{Gerochgboundary} itself, but also the bundle (or b-) boundary construction of Schmidt \cite{Schmidt} or the causal (or c-) boundary introduced by Geroch, Kronheimer and Penrose \cite{GKP}  (although there also idealized endpoints of complete timelike geodesics are attached),  cf.\ e.g.\ \cite{FloresHerreraSanchez}.\\[-0.5em]

Turning towards comparing our results with \cite{Chrusciel} we note that, while our proof of Theorem \ref{theo:uniquemaximalextension} largely follows the same overall strategy as \cite[Thm.\ 4.5]{Chrusciel} of constructing a larger extension from at least two given ones by taking unions and identifying appropriately,  \cite{Chrusciel} does this in the null bundle of $M$ whereas we work directly in the topological manifolds with boundary. 
In contrast to our arguments, \cite{Chrusciel} can work with geodesics up to and including the boundary and in particular has local uniqueness as in \cite[Prop.~3.5]{Chrusciel}. Having analogous tools in our setting would simplify parts of the proofs, e.g., in Lemma \ref{lem:pitilde_open} one could argue with geodesic uniqueness instead of using the  ''no intertwined geodesics'' assumption. On the other hand we do not have to show that our charts at the boundary are compatible nor have to construct a conformal metric that extends to the boundary.  
There is an interesting reformulation of this result in \cite[Thm.~5.3]{Chrusciel}, namely the existence of a unique future conformal boundary extension with
strongly causal boundary which is maximal in the class of future conformal boundary extensions with
strongly causal boundaries. It would be interesting to see if, for some sensible definition of strong causality for our boundaries, an analogous result remains available. Certainly Lemma \ref{lem:pitilde_open} seems amenable to a strong causality/non-imprisoning argument.\\[-0.5em]

Of course the bigger open questions are more conceptual. In Section \ref{sec:compatibility} we have discussed associating a (regular) future (g-)boundary extension to a given spacetime extension. Conversely one could ask if, given a regular future g-boundary extension, there is any hope of characterizing (intrinsically) when one can extend this further to a spacetime extension. Similarly, it is open if one could develop any  (in-)extendibility criteria ensuring the (non-)existence of future g-boundary extensions that do not come from spacetime extendibility and compatibility.
 In this sense the present article  can be considered a first starting point proposing a concept of regular future g-boundary extensions which, excluding pathological behavior of timelike geodesics in the original spacetime, naturally admits unique maximal elements, and many open questions remain to be explored.

\end{document}